\documentclass[numbook, envcountsect, envcountsame, envcountreset, runningheads, smallextended]{svjour3}
\usepackage{amssymb,amsmath,fixltx2e}
\usepackage{mathptmx}
\usepackage[normalem]{ulem}
\usepackage{graphicx}
\usepackage[arrow, matrix, curve]{xy}
\usepackage{txfonts}
\usepackage[usenames,dvipsnames]{color}
\usepackage{hyperref}

\usepackage{mathtools}									
\mathtoolsset{showonlyrefs=true}							
\MakeRobust{\eqref}

\newtheorem{assumption}{Assumption}

\newcommand{\Rpplus}{{\mathbb{R}_{\geq 0}}}
\newcommand{\Rplus}{{\mathbb{R}_{> 0}}}
\newcommand{\E}{\mathbb{E}}
\newcommand{\N}{\mathbb{N}}
\newcommand{\Q}{\mathbb{Q}}
\newcommand{\R}{\mathbb{R}}
\newcommand{\xstr}{x_{*}}
\newcommand{\ystr}{y_{*}}
\newcommand{\lam}{\lambda}
\newcommand{\q}{{\overrightarrow{q}}}

\renewcommand{\P}{\mathbb{P}}
\renewcommand{\rho}{\varrho}

\DeclareMathOperator{\dom}{dom\,}
\DeclareMathOperator{\epi}{epi\,}
\DeclareMathOperator{\hypo}{hypo\,}
\DeclareMathOperator{\co}{co\,}
\DeclareMathOperator{\supp}{supp \,}

\def\mathbbone{{\mathchoice {1\mskip-4mu\mathrm l}{1\mskip-4mu\mathrm l}
{1\mskip-4.5mu\mathrm l} {1\mskip-5mu\mathrm l}}}
\def\ind{\mathbbone}
\def\define{:=}
\def\JELname{{\bfseries JEL Subject Classification}}
\def\JELList#1{\par\addvspace\medskipamount{\rightskip=0pt plus1cm
\def\and{\ifhmode\unskip\nobreak\fi\ $\cdot$
}\noindent\JELname\enspace\ignorespaces#1\par}}

\smartqed

\title{Portfolio Optimization under Convex Incentive Schemes}
\author{Maxim Bichuch${}^*$  \and Stephan Sturm${}^*$}
\date{\today}

\institute{
              ${}^*$Department of Mathematical Sciences,
	      Worcester Polytechnic Institute,
	      100 Institute Road,
	      Worcester, MA 01609,
	      USA;
              \email{mbichuch@wpi.edu, ssturm@wpi.edu}.
	      Work partially supported by NSF grant DMS-0739195 while the authors were Postdoctoral Research Associates at Princeton University.
}

\journalname{Finance and Stochastics}

\begin{document}

\maketitle

\begin{abstract}
We consider the terminal wealth utility maximization problem from the point of view of a portfolio manager who is paid by an incentive scheme, which is given as a convex function $g$ of the terminal wealth. The manager's own utility function $U$ is assumed to be smooth and strictly concave, however the resulting utility function $U \circ g$ fails to be concave. As a consequence, the problem considered here does not fit into the classical portfolio optimization theory.  Using duality theory, we prove wealth-independent existence and uniqueness of the optimal portfolio in general (incomplete) semimartingale markets as long as the unique optimizer of the dual problem has a continuous law. In many cases, this existence and uniqueness result is independent of the incentive scheme and depends only on the structure of the set of equivalent local martingale measures. As examples, we discuss (complete) one-dimensional models as well as (incomplete) lognormal mixture and popular stochastic volatility models. We also 
provide a detailed analysis of the case where the unique optimizer of the dual problem does not have a continuous law, leading to optimization problems whose solvability by duality methods depends on the initial wealth of the investor.
\end{abstract}



\JELList{G11}

\section{Introduction}

Whereas classical portfolio theory studies utility maximization from the point of view of an investor, whose preferences are modeled by a concave utility function, in reality, portfolio management is commonly delegated to a fund manager. To increase the efficacy of the manager, he is often paid by an incentive scheme that depends on the performance of the fund he manages. Such a scheme can be composed, for example, of a fixed fee, some percentage of the fund, plus an additional reward, which consists of one (or a combination of several) call options on the fund. As a consequence, two differences to the classical setting arise. First, the utility function, under which the optimization is carried out, does not represent the preference structure of the investor (also called the {\em principal}), but rather the manager's (the {\em agent}'s) preference structure. Second, what is optimized under this utility function is not the terminal value of the fund itself, but rather some function of it, which depends on 
the specific incentive scheme.

The resulting optimization problem is, in general, no longer concave, and therefore does not fit into the classical setting first studied by Merton \cite{Merton}, who used a stochastic optimal control approach.  Specifically, Merton derived a Hamilton-Jacobi-Bellman (HJB) equation satisfied by the value function and found a closed form solution in the case of power utility. The drawback of this method -- namely that it requires the state process to be Markov -- can be overcome by using the fact that the processes dual to the portfolio processes are given via the set of equivalent local martingale measures, as pioneered by Karatzas, Lehoczky and Shreve \cite{KILS} and Pliska \cite{Pliska}.  A thorough study of the portfolio optimization problem in a general (incomplete) semimartingale market was conducted by Kramkov and Schachermayer \cite{KramSchach1}, \cite{KramSchach2}, Bouchard, Touzi and Zeghal \cite{BTZ} and others.

As pointed out, all of the above literature concentrates on the principal investor himself. The  problem becomes more involved, if the investor, rather than investing himself, delegates his money to a fund manager. The agent invests on the principal's behalf, in exchange for a fee schedule, which is based on the fund's performance at the final time $T$, and given by a function $g$ of the portfolio at terminal time.  We assume that the agent's utility function $U$ is smooth, strictly concave and has a domain bounded from below.  These assumptions allow for the classical examples of power and logarithmic utility (but not utility functions defined on the whole real line such as the exponential). The fee schedule function $g$ is assumed to be convex and dominated by an affine function -- i.e., its slope has to be bounded; without loss of generality we will assume that the maximal slope is $1$. The financial reasoning for these assumptions on $g$ is that we expect the manager's fees to increase as the fund's 
profit increases. Therefore, $g$ should be convex. The fund manager's utility, which results from his payoff, is hence a composition of the two functions, $\bar U \define U \circ g$, and may no longer be concave. Thus, the previously mentioned results are no longer applicable.  

The resulting problem is not well understood; the existing literature discusses mainly the question of whether such a compensation scheme leads the portfolio manager to take excessive risk. In \cite{Ross}, Ross discusses some conditions to make the agent more or less risk averse then the principal. Carpenter shows in \cite{carp} the existence of the fund manager's optimal portfolio in case of a utility function $U$ with constant relative risk aversion and a call option like fee schedule $g$ in a Brownian stock price model. In this setting, her analysis is generalized by Larsen \cite{Lars} into an agency problem, where the investor optimizes the resulting payoff over piecewise affine incentive schemes, which he might choose to offer the portfolio manager.

We want to point out that there is also a different approach to portfolio optimization under incentive schemes, in which the compensation is based on high-watermarks, i.e., the running maximum of the fund. Recent references to this compensation approach include \cite{ObGua}, \cite{JanSir}, \cite{PanWest}. In all of these papers the authors also assume a Brownian stock price model and solve the appropriate HJB equation.

In the present paper we will investigate the more fundamental problem of existence and uniqueness of an agent's optimal investment portfolio in a general semimartingale model. As noted above, the resulting fund manager's utility function $\bar U$ may not be concave. It is well-known that the solution is then to concavify $\bar U$, and solve the concavified problem instead. Even though this new utility is now concave, it is not necessarily strictly concave, nor does it necessarily satisfy the usual Inada condition at zero, both of which are needed in the classical utility maximization framework. Moreover, the smoothness of the concavified function is not clear a priori. Using a dynamic programming approach via HJB equation is -- at least in the straightforward way -- also not possible, since the concavified utility function can (and usually will) be affine in some parts, and hence finding the optimal portfolio becomes impossible. Thus, we have effectively to weaken the utility function requirement of Kramkov 
and Schachermayer \cite{KramSchach1}. Our approach is to use the more general framework of Bouchard, Touzi and Zeghal \cite{BTZ} and, by proving additional regularity of the concavified utility function, show the uniqueness of the dual optimizer. We are thus able to utilize the abstract framework of Bouchard, Touzi and Zeghal in a concrete setting, which is a rare feat (note, however, the exception of Seifried \cite{Seifried}, who discusses capital gains taxes in a complete market).

The next step is to develop sufficient conditions, broad enough to be of interest, for the solution of the concavified problem to be also the solution of the original problem. It turns out that a necessary and sufficient condition is that the corresponding unique dual optimizer has a continuous law (i.e., the distribution of the random variable has no atoms). A similar procedure can be found in a related paper by Carassus and Pham \cite{Carassus}, who consider a problem of portfolio optimization in a complete market with Brownian stock price, with a utility function created by two piecewise concave functions. We show that the condition of atomlessness holds, not only true in the classical Black-Scholes model with discounted stock price having nonzero drift, but also in two example classes of markets, independent of the initial capital of the fund and independent of the concrete incentive schemes: (1) complete one-dimensional It\^{o}-process models (such as local volatility models), and (2) incomplete 
lognormal mixture and stochastic volatility models (such as the popular correlated Hull-White, Scott and Heston models).

The practical consequence of this is that the agent shuns successfully away from any part of the domain where the concavified utility function is linear . However, he does this in a smooth way. The optimal terminal wealth has no atoms except possibly at zero (meaning that the fund manager jeopardizes the fund with a positive probability), and it is zero under any linear spot of the concavified utility function.

If the assumption on the non-atomic structure of the dual optimizers fails, we are still able to give an affirmative answer, albeit only for some initial capitals. In general, the fund manager's optimal wealth does not have to agree with the one calculated from the concavified problem, and even if it does, it does not have to be unique. As a note of caution, we present easy counterexamples that show that this method should not be implemented without proper conditions. We also give simple examples for our theorems, which conceptually present how the optimal portfolio can be explicitly calculated in a complete market setting.

The rest of this paper is organized as follows. In Section \ref{utility} we introduce the mathematical model of our delegated portfolio optimization problem and state our main results. The two following sections are devoted to examples illustrating our findings. In Section \ref{BSpower} we discuss in detail the case of power utility in a Black-Scholes market, highlighting the central importance of the distributional properties of the dual optimizer and investigating the problem from the point of view of the managers risk aversion. Section \ref{sec:example-diff} contains several complete and incomplete market models in which our assumptions hold true. The remaining sections are devoted to the more technical side of the problem. Section \ref{sec:prelim} provides the background on general results on smooth and non-smooth duality theory and discusses how they can put to work for our needs. Section \ref{genresult} contains the detailed proofs on the relationship of the conacavified and the dual problem. Section \
ref{sec:7} draws the conclusions for the original problem and contains the proof of the main theorem. Finally, Section \ref{sec:8} discusses the limitations of the main theorem and provides partial results for an atomic dual optimizer. The conclusions of our exposition are summarized in Section \ref{sec:conclusion}.

After finishing a first version of the present paper, we have learnt of the work of Reichlin \cite{Rei}, who studies the utility maximization problem for more general non-concave utility functions under a fixed pricing measure.

\section{Setting and Main Results}\label{utility}

We start by reviewing utility maximization in a general semimartingale framework and state our main results.  Assume that $S^i$, $i=1,\ldots , d$ is a d-dimensional, locally bounded semimartingale on a filtered probability space $(\Omega, \mathcal{F}, (\mathcal{F}_t)_{0 \leq t \leq T}, \P)$, representing discounted stock price processes; without loss of generality we assume $\mathcal{F}_T = \mathcal{F}$.  We focus on portfolio processes with initial capital $x$ and predictable and $S$-integrable hedging strategies $H$.  The value process of such a portfolio is then given by
\[
X_{t}^{x,H} = x + \int_{0}^{t} H_{s} \, dS_{s}, \qquad  0\le t\le T.
\]
Denote by $\mathcal{X}(x)$ the set of all nonnegative wealth processes  with initial capital $x$,
\begin{align}\label{procX}
\mathcal{X}(x) & = \Bigl\{ X \ge 0\, : \,   X_t = X_t^{x,H}  \mbox{ for some predictable and $S$-integrable strategy }H\Bigr. \nonumber  \\
& \phantom{=\Bigl\{ \Bigr.}\Bigl. \mbox{ for every } 0 \leq t \leq T \Bigr\}.
\end{align}
We refer to $\mathcal{X}(x)$ as the set of all {\em admissible wealth processes}.

We want to look at the portfolio optimization problem from the perspective of a portfolio manager, who is paid with incentives that depend on the performance of the portfolio at some future time $T>0$. In this article we allow the incentive scheme to be a function $g: \mathbb{R}_{\geq 0} \to \mathbb{R}_{\geq 0}$, nonconstant, nondecreasing, convex and with maximal slope $c>0$, i.e.,
\begin{equation}
\sup \bigcup_{x \geq 0} \partial g(x) \leq c.
\label{g-slope}
\end{equation}
We note that the agent's private capital can be absorbed into $g$ (if positive). To simplify the exposition, we will assume throughout this text that, without loss of generality, $c=1$. Setting $\bar{U} \define U \circ g$, the portfolio manager's utility maximization problem is
\begin{equation}\label{original}
u(x) \define \sup_{X \in \mathcal{X}(x)} \E\bigl[ \bar{U}\bigl(X_T\bigr)\bigr].
\end{equation}

\begin{assumption}\label{assumpt:finite}
To make the problem nontrivial, we assume that there exists at least some $x_0 > 0$ such that
\begin{equation*}
\sup_{X \in \mathcal{X}(x_0)} \E\bigl[U\bigl(X_T\bigr)\bigr] <\infty.
\end{equation*}
\end{assumption}

\begin{assumption}\label{equivlMart}
To preclude the possibility of arbitrage in the sense of {\em `free lunch with vanishing risk'} (for details see the work of Delbaen and Schachermayer, \cite{DS}) we assume that the set of equivalent local martingale measures is not empty,
\begin{equation*}
\mathcal{M}^{e} = \Bigl\{ \Q \, : \, \Q \sim \P, \, S \mbox{ is a local } \Q \mbox{-martingale}\Bigr\} \ne \emptyset.
\end{equation*}
\end{assumption}

\begin{assumption}\label{UassKS}
The fund manager's preferences are represented by a utility function $U: \Rplus \rightarrow \R$ (without loss of generality we assume $U(\infty) \define \lim\limits_{x\to\infty}U(x) > 0$). 
\begin{itemize}
\item[a)] We assume that $U$ is strictly increasing, strictly concave and continuously differentiable on $\Rplus$; we extend $U$ continuously to $\mathbb{R}_{\geq 0}$, allowing the value $-\infty$ at $0$;
\item[b)] The utility function satisfies the Inada-conditions
\begin{equation}\label{Inada}
U'(0)\define\lim\limits_{x\rightarrow0} U'(x)=\infty, \qquad  \qquad U'(\infty)\define\lim\limits_{x\rightarrow\infty} U'(x)=0,
\end{equation}
\item[c)]Moreover, it satisfies the asymptotic elasticity condition
\begin{equation}\label{AE}
AE(U) := \limsup_{x \to \infty} \frac{xU'(x)}{U(x)} < 1.
\end{equation}
\end{itemize}
\end{assumption}

These three standard assumptions of utility maximization problems (see, e.g., \cite{KramSchach1}) will be the standing assumptions for the rest of this paper.

Before introducing the dual problem, we recall some standard notions and notation of convex analysis. A function  $f \, : \, \mathbb{U} \subseteq \R \to [-\infty, \infty]$ defined on some convex domain $\mathbb{U}$ is called {\em convex} (respectively {\em concave}) if its epigraph (respectively hypograph)
\[
 \epi f := \bigl\{ (x,\mu) \in \mathbb{U} \times \R \, : \, f(x) \leq \mu \bigr\}, \qquad \hypo f := \bigl\{ (x,\mu) \in \mathbb{U} \times \R \, : \, f(x) \geq \mu \bigr\},
\]
is a convex set. The effective domain of a convex function $f$ is defined as
\[
 \dom f := \bigl\{ x \in \mathbb{U} \subseteq\R \, : \, f(x)< \infty \bigr\} .
\]
Similarly, for a concave function, we define its domain as the set of points in the pre-image not mapping to $-\infty$. Generalizing the usual notations from utility maximization problems in an obvious way, we define, for any function $f$ dominated by some affine function, its {\em convex conjugate} $f^{*}$ and its {\em biconjugate} $f^{**}$ by
\[
f^{*}(y) := \sup_{x \in \dom f} \Bigl(f(x)- xy\Bigr), \qquad f^{**}(x) := \inf_{y \in \dom f^*} \Bigl(f^*(y)+ xy\Bigr).
\]
Note that $f^{**}$ is the concavification of $f$, i.e., the hypograph of $f^{**}$ is the closed convex hull of the hypograph of $f$, $\hypo f^{**} = \overline{\co (\hypo f)}$. We note that $f^*$ is the convex conjugate of $-f(-\, \cdot \, )$ in the classical sense of convex analysis. We will use standard results  of convex analysis (cf., e.g., \cite{HUL}) with the obvious modifications without further notice.

We note that the function $\bar U$ is not necessarily concave, placing the problem \eqref{original} outside the standard setting of utility maximization. Instead of analyzing the non-concave problem \eqref{original} directly, we will first consider the concavified problem
\begin{equation}\label{concavified}
w(x) \define \sup_{W \in \mathcal{X}(x)} \E\bigl[ \bar{U}^{**}\bigl(W_T\bigr)\bigr].
\end{equation}
Similar to \cite{KramSchach1} we define the set of process dual to \eqref{procX} by
\begin{equation*}
\mathcal{Y}(y) := \Bigl\{ Y \geq 0 \, : \, Y_0 = y \mbox{ and } XY \mbox{ is a supermartingale for all } X \in \mathcal{X}(1)\Bigr\}.
\end{equation*}
It turns out then that both problems share the dual problem (see Theorem \ref{main} below), i.e.,
\begin{equation}\label{dualgen}
v(y) \define \inf_{Y \in \mathcal{Y}(y)} \E\bigl[ \bar{U}^{*}\bigl(Y_T\bigr)\bigr].
\end{equation}

In general the concavified utility function $\bar{U}^{**}$ will be neither strictly concave nor satisfy the Inada condition at $0$. Hence, we will have to rely on results for nonsmooth utility maximization (see Theorem \ref{BTZduality} for more details). We will see that Assumptions \ref{assumpt:finite}, \ref{equivlMart}, and \ref{UassKS} place us in a setting where we will be able to apply this  theorem.

Finally let
\[
\beta := \inf \{ x >0 : \bar{U} (x) >-\infty\} \in [0,\infty).
\]
The following are the main theorems of this paper. Theorem \ref{main} establishes the duality relationship between 
$v$ and $w$, and relates $\hat W_T(x)$ and $\hat Y_T(y)$, the optimizers of the problems \eqref{concavified} and \eqref{dualgen}, respectively. Theorem \ref{general} provides conditions under which $\hat X_T(x)$ and $\hat W_T(x)$ -- the optimizers of the problems \eqref{original} and \eqref{concavified}, respectively, are the same.

\begin{theorem}\label{main}
For the utility optimization problem under a convex incentive scheme $g$ it holds that
\begin{itemize}
\item[a)] The functions $u$ and $w$ are finite on $(\beta, \infty)$ as is $v$ on $\mathbb{R}_{> 0}$, and $v = w^*$. Moreover, $v$ is strictly convex on the whole domain $(0,\infty)$ if $U(0)=-\infty$. Otherwise, there exists some $\delta \in (0,\infty]$ such that $v$ is convex on the interval $(0,\delta)$ and constant $\bar{U}^{**}(0)$ on $[\delta, \infty)$.  The function $w$ is continuously differentiable on $(\beta,\infty)$, and concave.
\item[b)] The optimizer $\hat{Y}(y)$ of the dual problem \eqref{dualgen} exists for every $y>0$ and is a.s. unique on $\bigl(0,\bigl(\bar{U}^{**}\bigr)'(0)\bigr)$.
\item[c)] For $x > \beta $ there exists an optimizer $\hat{W}	(x)$ of the concavified problem \eqref{concavified} that satisfies
\[
\hat{W}_T(x) \in - \partial \bar{U}^{*}\bigl(\hat{Y}_T(y)\bigr)
\]
for $y = w'(x)$ such that $\hat{W}(x) \in \mathcal{X}(x)$ and $\hat{W}(x)\hat{Y}(y)$ is a uniformly integrable martingale.
\item[d)] Additionally we have
\[
v(y) = \inf_{\Q \in \mathcal{M}^e} \E\biggl[ \bar{U}^*\biggl(y \frac{d\Q}{d\P}\biggr)\biggr],
\]
however the infimum is in general not attained in $\mathcal{M}^e$.
\end{itemize}
\end{theorem}

\begin{theorem}\label{general}
Assume that for every $y \in \bigl(0,w'(\beta)\bigr]$ the terminal value of the dual optimizers $\hat{Y}(y)$ has a continuous cumulative distribution function. Then 
\begin{itemize}
\item[a)] The optimizer $\hat{W}(x)$ for the concavified problem \eqref{concavified} is unique for every $x>\beta$.
\item[b)] For every $x>\beta$ there exists a solution $\hat{X}(x)$ of the original problem \eqref{original} and this solution is unique. It coincides also with $\hat{W}(x)$, the solution of the concavified problem \eqref{concavified}.
\end{itemize}
\end{theorem}
At a first glance the condition that the distribution of the dual optimizer has no atoms seems quite abstract and hard to check. Therefore, we present a sufficient condition for no atoms, in terms of equivalent local martingale measures, which can be checked more easily in some concrete models.

\begin{proposition}\label{suffCrit}
Assume that the laws of the Radon-Nikod\'{y}m derivatives $Z_T = \frac{d\Q}{d\P}\vert_{\mathcal{F}_T}$, $\Q \in \mathcal{M}^e$, are uniformly absolutely continuous with respect to the Lebesgue measure $\lambda$ on $\Rplus$ (i.e., the densities $\frac{d\P \circ Z_T^{-1}}{d\lambda}$ are uniformly integrable). Then the terminal value of the optimizer $\hat{Y}(y)$ of the dual problem \eqref{dualgen} has a continuous law. 
\end{proposition}

While this assumption it is quite restrictive, it is the only one we know which works in general without having a priori knowledge of the maximizer. It is in particular satisfied in the Black-Scholes model with nonzero drift. We will show in section \ref{sec:example-diff} that these assumptions are also satisfied in other incomplete market models, such as lognormal mixture models. In more general models -- as stochastic volatility models, see also section \ref{ex:stochVol} -- one can nevertheless derive the result, if one has some a priori knowledge about the optimizer, essentially depending on the measurability properties of the Sharpe ratio.

\section{Examples around the Black-Scholes model}\label{BSpower}

We first present our findings for power utility maximization in the Black-Scholes model with an incentive $g$ of call option type: $g(x)=\lambda(x-k)^{+}$. This \emph{setting} not only allows us to connect our results to previous work \cite{carp} and provide explicit solutions, but it also allows us to illustrate the degeneracy if the Sharpe ratio vanishes -- producing a purely atomic Radon-Nikod\'{y}m derivative.  Another benefit of studying this setting is that it allows us to analyze the situation from the point of view of the (relative) risk aversion of the manager. Specifically, we address the issue of the optimal relation between the number of options and the value of the strike among all those producing the same relative risk aversion (compare this to \cite{Ross}). 

\begin{example}\label{counter}
Assume that the discounted stock price is modeled by
\[
S_t = \exp{\Bigl(\sigma W_t +\bigl(\mu-\sigma^2/2\bigr)t\Bigr)}, \qquad \mu\ge 0, \, \, \sigma >0,
\]
for some Brownian motion $W$ generating the filtration $(\mathcal{F}_t)$. These stock price dynamics, together with the riskless num\'{e}raire, describe a complete market. The set of all equivalent local martingale measures is hence the singleton $\mathcal{M}^e=\{\Q\}$, where the measure $\Q$ is given by the Radon-Nikod\'{y}m density $Z_T = \frac{d\Q}{d\P}\vert_{\mathcal{F}_T} = \exp{\bigl\{ - \theta W_T - \frac{\theta^2}2 T\bigr\}}$ with market price of risk  $\theta \define \frac{\mu}{\sigma}$ and $W^{\Q}_t \define W_t + \theta t$ is a $\Q$-Brownian motion. Furthermore, let the incentive scheme be given by $g(x) =\lambda(x-k)^{+}$, $k>0$, $0<\lambda<1$.

The portfolio manager's utility will be given by the function be $U$. We will now consider two cases, $U(0) =-\infty$ and $U(0)>-\infty$. In the second case, we will assume without loss of generality that $U(0)=0$ and find
\begin{align*}
\bar{U}(x) & = \left\{ \begin{array}{ll} 0 & 0 \leq x \leq k,\\ U\bigl(\lambda(x-k)\bigr) & x > k,\end{array} \right. \quad
 \bar{U}^{**}(x) = \left\{ \begin{array}{ll} \ystr x & 0 \leq x \leq \xstr,\\ U\bigl(\lambda(x-k)\bigr)  & x > \xstr,\end{array} \right.  \\
 \bar{U}^{*}(y) & = \left\{ \begin{array}{ll} U^{*}(y/\lambda)-ky & 0 < y \leq\ystr,\\ 0 & y >  \ystr,\end{array} \right.
\end{align*}
where $\xstr$ is the solution of $\lambda\xstr U'\bigl(\lambda(\xstr-k)\bigr)=U\bigl(\lambda(\xstr-k)\bigr)$, and $\ystr = \lambda U'\bigl(\lambda(\xstr-k)\bigr)$.
In the former case
\[
\bar{U}^{**}(x) = \bar{U}(x) =  \left\{\begin{array}{ll} - \infty & 0 < x \leq k, \\ U\bigl(\lambda(x-k)^{+}\bigr) & x > k, \end{array}\right. \quad
 \bar{U}^{*}(y) = U^{*}(y/\lambda)-ky, \;\; y > 0.
\]
To make our example more computationally tractable, we will focus on the power utility case $U(x) =\frac{x^p}{p}$ with $0<p<1$. Here $U^{*}(y) = \frac{1-p}{p}y^{\frac{p}{p-1}}$, $\xstr = \frac{k}{1-p}$, and $\ystr = \lambda^p\bigl(\frac{p}{1-p}k\bigr)^{p-1}$, and thus
\begin{align}
\bar{U}(x) & = \left\{ \begin{array}{ll} 0 & 0 \leq x \leq k,\\ \frac{(\lambda(x-k))^p}{p} & x > k,\end{array} \right. \quad
 \bar{U}^{**}(x) = \left\{ \begin{array}{ll} \lambda^p x( \frac{pk}{1-p})^{p-1}  & 0 \leq x \leq \frac{k}{1-p},\\ \frac{\bigl(\lambda(x-k)\bigr)^p}{p}   & x > \frac{k}{1-p},\end{array} \right. \nonumber \\
 \bar{U}^{*}(y) & = \left\{ \begin{array}{ll} \frac{1-p}{p}\Bigl(\frac{y}{\lambda}\Bigr)^{\frac{p}{p-1}}-ky & 0 < y \leq\ystr,\\ 0 & y >  \ystr.\end{array} \right.
 \label{eq:U}
\end{align}

For the illustration in Figure \ref{fig:numerical3} we assumed $p=\frac12,\lambda=\frac14, k=3,U(x) = 2\sqrt{x}$. Then it is easily seen that $\xstr = 6,~\ystr=\frac{\sqrt{3}}{6}$ and
\begin{align}
\bar{U}(x) & = \left\{ \begin{array}{ll} 0 & 0 \leq x \leq 3,\\ \sqrt{x-3} & x > 3,\end{array} \right. \quad
 \bar{U}^{**}(x) = \left\{ \begin{array}{ll} \frac{\sqrt{3}}{6} x & 0 \leq x \leq 6,\\ \sqrt{x-3} & x > 6,\end{array} \right.  \nonumber \\
 \bar{U}^{*}(y)& = \left\{ \begin{array}{ll} \frac{1}{4y} -3y & 0 < y \leq\frac{\sqrt{3}}{6},\\ 0 & y >  \frac{\sqrt{3}}{6}.\end{array} \right.
 \label{eq:U-example}
\end{align}

 \begin{figure}[htbp]
 \begin{center}
 \includegraphics[width=0.45\linewidth]{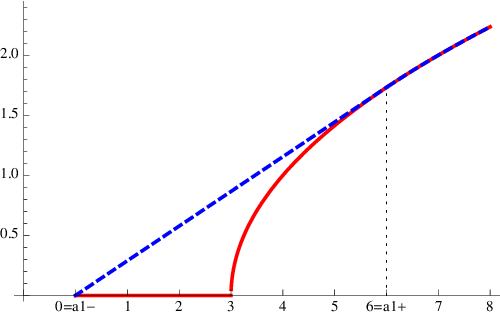}\qquad\qquad  \includegraphics[width=0.45\linewidth]{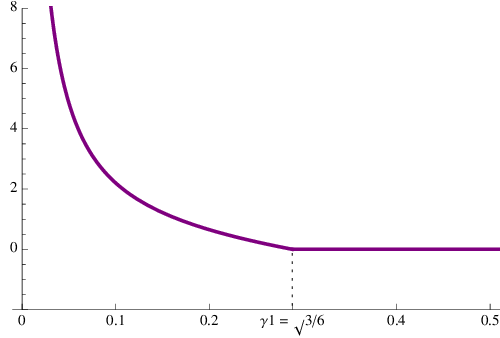}
 \caption{Left: Composed utility function $\bar{U}$ of equation \eqref{eq:U-example} and its concavification $\bar{U}^{**}$. Right: Dual utility function $\bar{U}^*$.}
 \label{fig:numerical3}
 \end{center}
 \end{figure}

The non-atomicity condition on the dual optimizer (here just the Radon Nikod\'{y}m derivative $\frac{d\Q}{d\P}$) implies that we have to consider two distinct cases: either $\mu$ is positive (equivalently, we could assume $\mu$ negative), or $\mu=0$. The reason for this distinction lies in the fact that with $\theta=0$, the original measure $\P$ is already the (unique) risk-neutral measure $\Q$. In other words, $Z_T = \frac{d\Q}{d\P}\vert_{\mathcal{F}_T} = 1$. Hence, the random variable $Z_T$ has an atom of mass one at $1$. In the other case the random variable $Z_T=\frac{d\Q}{d\P}\vert_{\mathcal{F}_T}$ has a (smooth) density. 
\subsection{Case 1: $\theta>0$}\label{ex:nonZeroDrift}
We proceed first with the case in which the random variable $Z_T$ has a density. Tedious, but straightforward stochastic calculus reveals the following results. The dual value function is given by 
\begin{align*}
v(y) & = \E \bigl[\bar{U}^*(yZ_T)\bigr] = \frac{1-p}p\Bigl(\frac{y}{\lambda}\Bigr)^{\frac{p}{p-1}}e^{\frac{p}{(1-p)^2}\frac{\theta^2 T}{2}} \Phi\bigl(d_{+} \bigr) -ky\Phi\bigl(d _{-}\bigr),\\
d_{+} & :=\frac{\log\ystr-\log y}{\theta\sqrt{T}}+ \Bigl(\frac12+\frac{p}{1-p}\Bigr)\theta\sqrt{T},\\
d_{-} & :=\frac{\log\ystr-\log y}{\theta\sqrt{T}}-\theta\sqrt{T},
\end{align*}
where $\Phi$ is the cumulative distribution function of the normal distribution. Thus, $v$ is continuously differentiable and strictly concave on the whole real line. Therefore, using the fact that $w^{*}=v$, the function $w$ is also strictly concave and continuously differentiable on $\Rplus$. For the terminal value of the optimizer we obtain (we can use almost everywhere defined derivatives since the law of $\hat{Y}_T(y) = y\frac{d\Q}{d\P}$ has no atoms)
\begin{align*}
\hat{X}_T(x) & = \hat{W}_T(x) = -\bigl(\bar{U}^*\bigr)' \Bigl(\hat{Y}_T\bigl(w'(x)\bigr)\Bigr) \\
& = \left(\Biggl(\frac{\lambda^p e^{ \theta W_T + \frac{\theta^2 T}{2} }}     {w'(x)}\Biggr)^{\frac1{1-p}}+ k\right)\mathbbone_{\bigl\{W_T > \frac{1}{\theta}\bigl(  \log w'(x)-\log \ystr  \bigr) - \frac{\theta T}{2}\bigr\}}.
\end{align*}
To compute the optimal strategy we simply observe that
\begin{align*}
f(t,z) & = \frac1{\sqrt{2\pi  (T-t)   }}\int_{\frac{1}{\theta}\bigl(  \log w'(x)-\log \ystr  \bigr) + \frac{\theta T}{2}}^\infty e^{ -\frac{(y-z)^2}{2(T-t)} } \left(\Biggl(\frac{\lambda^p e^{ \theta y - \frac{\theta^2 T}{2} }}     {w'(x)}\Biggr)^{\frac1{1-p}}+ k\right)  dy
\end{align*}
solves the (reverse) heat equation on $[0,T) \times \R$ with terminal condition
\[
f(T,z) = \left(\Biggl(\frac{\lambda^p e^{ \theta z - \frac{\theta^2 T}{2} }}{w'(x)}\Biggr)^{\frac1{1-p}}+ k\right)\mathbbone_{\bigl\{z > \frac{1}{\theta}\bigl(  \log w'(x)-\log \ystr  \bigr) + \frac{\theta T}{2}\bigr\}}
\]
satisfying $f(T, W_T^\Q) = \hat{X}_T(x)$. Then it follows from It\^{o}'s formula that 
\[
\hat{X}_T(x) = x + \int_0^T f_z(t,W_t^{\Q})dW_t^{\Q} = x+ \int_0^T \frac{f_z(t,W_t + \theta t)}{\sigma S_t}dS_t. 
\]
Thus, the optimal strategy (in terms of cash invested in stock) is simply $\frac{f_z(t,W_t + \theta t)}{\sigma}$.

Explicitly, we derive
\begin{align*}
f(t,z) = & \Biggl(\frac{\lambda^p } {w'(x)}\Biggr)^{\frac1{1-p}} e^{ \frac{\theta z}{1-p} +\frac{\theta^2}{2(1-p)}\bigl(-T + \frac{T-t}{1-p} \bigr) } \\
& \phantom{=} \cdot \Phi\Biggl( - \frac{\log w'(x)-\log \ystr}{\theta\sqrt{T-t}} - \frac{\theta T}{2\sqrt{T-t}}  +\frac{z}{\sqrt{T-t}} + \frac{\theta\sqrt{T-t}}{1-p }\Biggr)\\
& + k \Phi \Biggl( - \frac{\log w'(x)-\log \ystr}{\theta\sqrt{T-t}} - \frac{\theta T}{2\sqrt{T-t}} + \frac{z}{\sqrt{T-t}}\Biggr).
\end{align*}
Thus, the optimal strategy in terms of money invested in stock is given by
\begin{align*}
H_t(x) = & \frac{\theta}{(1-p)\sigma } \Biggl(\frac{\lambda^p } {w'(x)}\Biggr)^{\frac1{1-p}} e^{ \frac{\theta (W_t+\theta t) }{1-p}+ \frac{\theta^2}{2(1-p)}\bigl(-T + \frac{T-t}{1-p} \bigr) } \\
& \phantom{=} \cdot\Phi \Biggl( - \frac{\log w'(x)-\log \ystr}{\theta\sqrt{T-t}} - \frac{\theta T}{2\sqrt{T-t}}  +\frac{W_t+\theta t}{\sqrt{T-t}} + \frac{\theta\sqrt{T-t}}{1-p }\Biggr)\\
&+\frac{1}{\sigma\sqrt{T-t} } \Biggl(\frac{\lambda^p } {w'(x)}\Biggr)^{\frac1{1-p}} e^{ \frac{\theta (W_t+\theta t) }{1-p} +\frac{\theta^2}{2(1-p)}\bigl(-T + \frac{T-t}{1-p} \bigr) } \\
& \phantom{=} \cdot\varphi \Biggl( - \frac{\log w'(x)-\log \ystr}{\theta\sqrt{T-t}} - \frac{\theta T}{2\sqrt{T-t}}  +\frac{W_t+\theta t}{\sqrt{T-t}} + \frac{\theta\sqrt{T-t}}{1-p }\Biggr)\\
& + \frac{k}{\sigma\sqrt{T-t} }  \varphi \Biggl( - \frac{\log w'(x)-\log \ystr}{\theta\sqrt{T-t}}  - \frac{\theta T}{2\sqrt{T-t}}  +\frac{W_t+\theta t}{\sqrt{T-t}} \Biggr).
\end{align*}
We note that this can also be derived more generally by using \cite[Theorem 6.2]{OconeKaratzas} and observing that their condition (6.7) is always satisfied in our case, as asymptotic elasticity implies the inequalities of \cite[Lemma 6.3]{KramSchach1} (where we only have to replace the derivatives by suprema of subdifferentials). 

On a more practical side, we would like to investigate the optimization problem conditional on the portfolio manager's risk aversion. It turns out that the proper concept for this issue is to look on relative risk aversion (RRA) and to apply this also to the dual value function.
Using the fact that $(w')^{-1} = v'$, it follows that $v''(y) = \frac1{w''(x)}$ for $y=w'(x)$, and we can compute the relative risk aversion
\[
RRA_{v}(y) = -\frac{yv''(y)}{v'(y)} = -\frac{w'(x)}{xw''(x)} = \frac{1}{RRA_{w}(x)} =\frac{1}{RRA_{w}(v'(y))}. 
\]
Using the fact the $d_{-}'=d_{+}'=-\frac{1}{y\theta\sqrt{T}}$, and that $\Phi ''(x) = -x\Phi '(x)$, we obtain
\begin{align*}
& \phantom{=} RRA_{v}(y) \\ & = \frac{\lambda^{\frac{p}{1-p}} y^{\frac{1}{p-1}} e^{\frac{p}{(1-p)^2}\frac{\theta^2 T}{2}}
 \Bigl( \frac{\Phi(d_{+})}{1-p}+\frac{\Phi'(d_{+}) }{p\theta\sqrt{T}} + \frac{\Phi'(d_{+})}{\theta\sqrt{T}} - \frac{1-p}{p}\frac{d_{+}\Phi'(d _{+}) }{\theta^2 T} \Bigr) + \frac{k\Phi'(d_{-}) } {\theta\sqrt{T}} \Bigl( 1 + \frac{d_{-}}{\theta\sqrt{T}}\Bigr)}{-\lambda^{\frac{p}{1-p}} y^{\frac{1}{p-1}} e^{\frac{p}{(1-p)^2}\frac{\theta^2 T}{2}}
\Bigl( \Phi(d _{+})+\frac{1-p}{p}\frac{\Phi'(d _{+}) }{\theta\sqrt{T}}\Bigr) +\frac{k}{\theta\sqrt{T}} \Phi '(d_{-})- k\Phi(d_{-})}\\
& =\frac{\lambda^{\frac{p}{1-p}} y^{\frac{1}{p-1}} e^{\frac{p}{(1-p)^2}\frac{\theta^2 T}{2}}
\Bigl( \frac{\Phi(d _{+})}{1-p}+\Phi'(d _{+})\bigl( -\frac{1-p}{p}\frac{\log \ystr-\log y}{\theta^3T^{\frac32}} + \frac12\frac{1+p}{p}\bigr)\Bigr) + \frac{k\Phi'(d_{-}) } {\theta\sqrt{T}} \frac{\log\ystr-\log y}{\theta^2T}}{-\lambda^{\frac{p}{1-p}} y^{\frac{1}{p-1}} e^{\frac{p}{(1-p)^2}\frac{\theta^2 T}{2}}
\Bigl( \Phi(d_{+})+\frac{1-p}{p}\frac{\Phi '(d_{+}) }{\theta\sqrt{T}}\Bigr) +\frac{k}{\theta\sqrt{T}} \Phi '(d_{-})- k\Phi(d_{-})      }.
\end{align*}

\begin{figure}[htbp]
\begin{center}
\includegraphics[width=.45\linewidth]{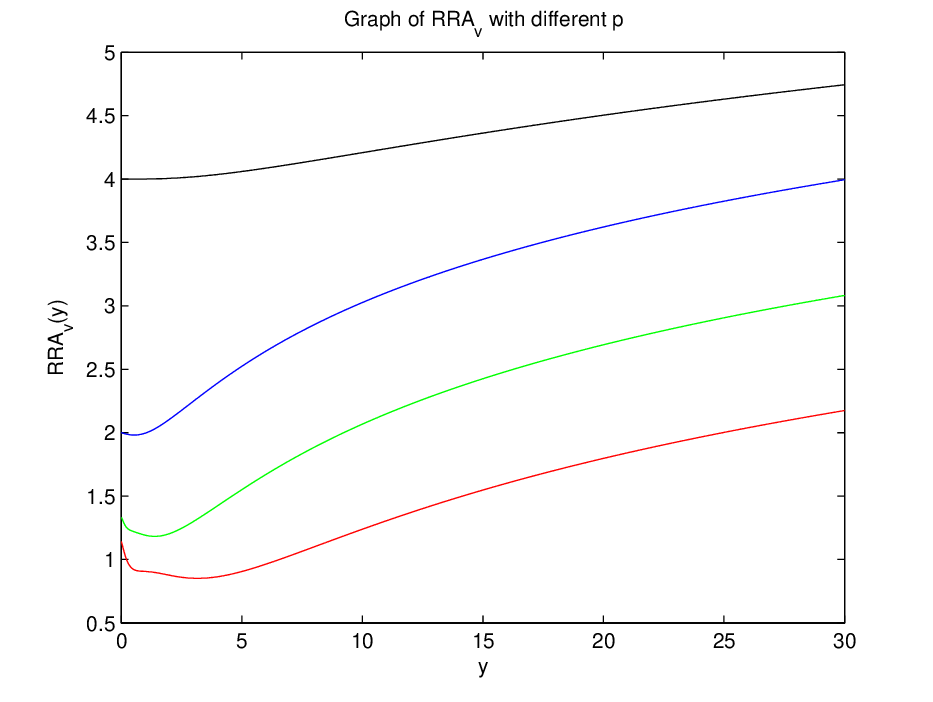}\qquad\qquad  \includegraphics[width=0.45\linewidth]{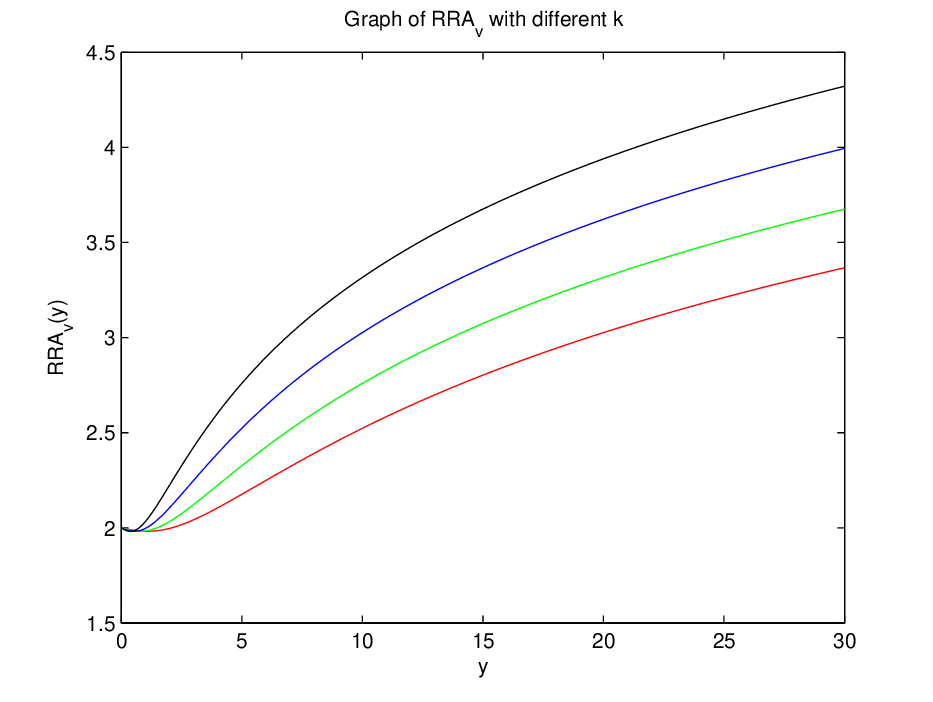}
\caption{Left: $RRA_{v}$ for values of $p=\frac18,\frac14,\frac12,\frac34$ in red, green, blue and black. Right: $RRA_{v}$ for values of $k=\frac14,\frac12,1,2$ in red, green, blue and black.}
\label{fig:numerical4}
\end{center}
\end{figure}

For the rest of this subsection, to highlight the dependency on $k$ and $\lambda$, we will write $\bar U^{*}_{k,\lam}(y)$, $u_{k,\lam} (x)$, $v_{k,\lam}(x)$, $w_{k,\lam}(x)$ for the concavified dual utility functions, the value function, its dual, and the concavified value function, respectively. 

Notice that if we parametrize $k$ and $\lam$ by $\alpha>0$ in the way that $\alpha \mapsto \bigl(k(\alpha), \lam(\alpha)\bigr) = \Bigl(\alpha \kappa, \alpha^{\frac{1-p}{p}} l\Bigr)$, then $\ystr= \lambda^p(\alpha)\Bigl(\frac{p}{1-p}k(\alpha)\Bigr)^{p-1}  = l^p \bigl(\frac{1-p}{p \kappa}\bigr)^{1-p}$ does not depend on $\alpha$. The same is true for $d_{\pm}$, and it follows that $v^{k(\alpha),\lam(\alpha)}(y)=\alpha v^{k(1),\lam(1)}(y)$. This is not surprising, since the same scaling property holds for the concavified dual utility function $\bar U^{*}_{k(\alpha),\lam(\alpha)} = \alpha \bar U^{*}_{k(1),\lam(1)}$.
Finally, we conclude that $RRA_{v^{k(\alpha),\lam(\alpha)}}(y)$ does not depend on $\alpha$. That is, the relative risk aversion of the fund manager is does not change when his compensation scheme is scaled in the above way.
Additionally, it is easily seen that
\[
w^{k(\alpha),\lam(\alpha)}(x) =  \inf_{y >0} \Bigl(v^{k(\alpha),\lam(\alpha)}(y)+ xy\Bigr) = \alpha\inf_{y >0} \Bigl(v^{k(1),\lam(1)}(y)+ \frac{x}{\alpha}y\Bigr) = \alpha w^{k(1),\lam(1)}\Bigl(\frac{x}{\alpha}\Bigr).
\]
Note that the family $\bigl\{ v^{k(\alpha),\lam(\alpha)}(y) \bigr\}_{\alpha>0}$ includes all the functions (up to an additive constant) that have the same relative risk aversion as the original dual function $v$. 

This leads finally to the following questions. Among the incentive schemes with same relative risk aversion $RRA_w(x)$, is there is one that is optimal from the manger's point of view?  If so, how it can be characterized? It turns out that the answer to the former question is affirmative if there is some $c^*$ such that the elasticity of $w^{k(1),\lam(1)}$ is equal to one: $E(w^{k(1),\lam(1)})(c^*)=1$. The elasticity of a utility function $U$ being defined as
\begin{equation}\label{E}
E(U)(c) := \frac{cU'(c)}{U(c)},~c>0,
\end{equation}
(compare this with the definition of the asymptotic elasticity in \eqref{AE}). Indeed, any solution $\alpha$ to $\frac{\partial}{\partial \alpha}w^{k(\alpha),\lam(\alpha)}(x) = w^{k(1),\lam(1)}\bigl(\frac{x}{\alpha}\bigr) - \frac{x}{\alpha}(w^{k(1),\lam(1)})'\bigl(\frac{x}{\alpha}\bigr)=0$ is precisely a solution to $E\bigl(w^{k(1),\lam(1)}\bigr)\bigl(\frac{x}{\alpha}\bigr)=1$. Hence, an optimal solution is characterized via $\alpha = \frac{x}{c^*}$, and thus $g(x) = \bigl(\frac{x}{c^*}\bigr)^\frac{1-p}{p}\bigl(x - \frac{x}{c^*}\bigr)^+$. Moreover, because $w^{k(\alpha),\lam(\alpha)}$ is concave this is a maximum. In the case that $E(w^{k(\alpha),\lam(\alpha)})<1$ on $\Rplus$ there is no optimal $\alpha$, as the manager's expected utility increases as $\alpha$ tends to infinity since $\frac{\partial}{\partial \alpha}w^{k(\alpha),\lam(\alpha)}>0$.

\subsection{Case 2: $\theta=0$}\label{ex:zeroDrift}
We now consider the second case, in which $\theta=0$. We remind the reader that we are still assuming that $U(x) =\frac{x^p}{p},~0<p<1$. Also, for future reference, note that 
\begin{equation}
\bar{U}^{**}(x) > \bar{U}^{}(x),~x\in(0,\xstr).
\label{eq:inequality-U}
\end{equation}
Indeed, this holds for $x\in(0,k]$, where from \eqref{eq:U} $\bar{U}^{**}(x) = \lambda^p x( \frac{pk}{1-p})^{p-1} >0$, and for $x\in(k,\xstr)$ it is true, because  $\bigl( \bar{U}^{**}\bigr)'(x) < \bigl( \bar{U}^{}\bigr)'(x)$ and $\bar{U}^{**}(\xstr) = \bar{U}^{}(\xstr)$. As mentioned above this case is different from Case 1 as we have now $\mathcal{M}^e=\{\P\}$ with $Z_T = \frac{d\P}{d\P}\vert_{\mathcal{F}_T} = 1$. It follows that the dual value function is given by $v(y) = \bar{U}^*(y)$ and thus, using the fact that $w^{*}=v$, we have for the concavified problem $w(x) = \bar{U}^{**}(x)$.

We first consider the case when $x\in(0,\xstr)$. We have from duality that $w(x) =  \bar{U}^{**}(x)$. This optimum is of course attained by the trivial strategy $H \equiv 0$ yielding the optimal wealth process $\hat{W}(x) \equiv x$ for the concavified problem. However, plugging this into the original problem yields $\E\bigl[\bar{U}\bigl(\hat{W}_T(x)\bigr)\bigr] = \bar{U}^{}(x)$. Thus, $\hat{W}_T(x)$ is an optimizer for the concavified problem, but yields a smaller value for the original problem.  Moreover, $\hat{W}_T(x)$ is even not an optimizer for the primal problem, as we will show that $\bar{U}^{}(x)<u(x) =w(x) = \bar{U}^{**}(x)$.

In this example, a way around this problem can be seen by thinking in terms of investment strategies. Not only does the trivial strategy $H \equiv 0$ lead to the optimum for the concavified problem, but so does every strategy with terminal value $\hat{W}_T(x)$ satisfying $\supp \hat{W}_T(x) \subseteq [0,\xstr]$, since in the interval $[0,\xstr]$ the concavified utility function $\bar{U}^{**}$ is linear.  Therefore, by the martingale property of the wealth process under $\P$, we have $\E \bigl[\bar{U}^{**} \bigl(\hat{W}_T(x)\bigr)\bigr] = \bar{U}^{**}(x)$. However, any strategy yielding a terminal value $\hat{W}_T(x)$, which has some support in $(\xstr,\infty)$ is clearly not optimal by the strict concavity of the concavified utility function. Finally, a strategy that maximizes, not only the concavified problem, but also yields the same value for original problem, has to satisfy $\supp \hat{X}_T(x) = \supp \hat{W}_T(x) = \{0,\xstr\}$ since $\bar{U} < \bar{U}^{**}$ on $(0,\xstr)$ 
by \eqref{eq:inequality-U}.

The existence of this strategy follows from a simple application of martingale representation theorem. Indeed, fix $a$ such that $\Phi(a/\sqrt{T})=\frac{x}{\xstr}$. Then, by the martingale representation theorem, the random variable $ \xstr\ind_{\{W_T<a\}}$ has the representation $\xstr\ind_{\{W_T<a\}} = x + \int_0^T H_sdW_s$, where $H\in L^2([0,T]\times\Omega)$. Thus, $\xstr\ind_{\{W_T<a\}} = x + \int_0^T \frac{H_s}{\sigma S_s}dS_s$. \footnote{We thank an anonymous referee for pointing out this straightforward existence proof.}

It turns out that one can easily explicitly construct the optimal strategy by using a strategy similar to the classical doubling strategy in the Black-Scholes model. However, contrary to the classical dubbling strategy our strategy will be admissible. Define the strategy $H_t = \frac{1}{\sigma S_t \sqrt{T-t}}$, which gives rise to the value process
\[
X_t^{1,H} = x +  \int_0^t \frac{dS_s}{\sigma S_s\sqrt{T-s}} = x +  \int_0^t \frac{dW_s}{\sqrt{T-s}}.
\]
We note that $X^{1,H}$ is a local martingale with quadratic variation process
\[
\Bigl\langle X^{1,H} \Bigr\rangle_t = \int_0^t \frac{ds}{T-s} = \log{\frac{T}{T-t}}, 
\]
hence, it is a time changed Brownian motion $X_t^{1,H} = x+\tilde{W}_{\log{\frac{T}{T-t}}}$. Defining now the stopping time $\tau := \inf \{t \geq 0 \, : \, X_t^{1,H} \notin [0,\xstr]\}$ we can see that we have for the stopped strategy $H_t^\tau = \frac{1}{\sigma S_t^\tau \sqrt{T-t}}$,
\[
X_t^{1,H^\tau} = x +  \int_0^t \mathbbone_{\{s \leq \tau \}}\frac{dW_s}{\sqrt{T-s}} = x+\tilde{W}_{\log{\frac{T}{T-t}}}^\tau. 
\]
Thus, the process $X^{1,H^\tau}$ hits either $0$ or $\xstr$ before time $T$ a.s. and the stopped process at terminal time, $X_T^{1,H^\tau}$, is hence almost surely concentrated on $\{0,\xstr\}$. Thus, $H^\tau$ is indeed a strategy which yields the optimum. 

This example, specifically the treatment of the case $\theta=0$ in section \ref{ex:zeroDrift}, reveals yet an other interesting fact. While the dual optimizer $\hat{Y}_T(y) =y$ is purely atomic for ever $y>0$, for $x\ge \xstr$ it nevertheless follows that $w(x) =  \bar{U}^{**}(x) = \bar{U}(x)$ is reached also by the trivial strategy $H \equiv 0$.  However, in this case, the solution of the concavified and the original problem coincide. This means that the condition of the atomlessness of the dual optimizer is not a necessary one, at least as one does not require an existence result which is independent of the initial capital.
\end{example}

\section{Examples of Models in the Class of It\^{o} Process}\label{sec:example-diff}

We now want to illustrate that Theorems \ref{main} and \ref{general} not only hold in Black-Scholes type markets, but also in many complete and incomplete markets, where the stock price process is given by an It\^{o} process.  First we will consider complete market models given by one-dimensional It\^{o} processes and prove a general sufficient condition in terms of Malliavin differentiability, which can be applied, e.g., to local volatility models. Then we show that some classes of incomplete market models, such as the lognormal mixture models of Brigo and Mercurio \cite{Brigo}, satisfy the conditions of Proposition \ref{suffCrit}. Additionally, we show that, under certain assumptions, stochastic volatility models satisfy directly the conditions of Theorem \ref{general}. We also provide examples of well-known models by Hull-White, Heston and Scott satisfying those assumptions.

\begin{example}\label{ex1}
{\bf (One dimensional diffusion models)}: Let $W$ be a one-dimensional
\footnote{Generalization to the  multi-dimensional case is straightforward. However, to make the exposition more tractable, we stay in one dimension.}
 Brownian motion, defined on some probability space $(\Omega, \mathcal{F}, P)$.  Denote by $\bigl(\mathcal{F}_t^{W}\bigr)$ the filtration generated by the Brownian motion, augmented by all $\P$-negligible sets (as usual, we assume without loss of generality that $\mathcal{F}_T^{W}=\mathcal{F}$). Additionally, let $\mathcal{B}([0,t])$ denote the Borel-$\sigma$-field on the interval $[0,t]$. Let the stock price process given by
 \begin{equation}
 dS_t = \mu_t S_t \, dt + \sigma_t S_t \, dW_t, \qquad S_0 =s,
 \label{eq:S}
 \end{equation}
 where $\mu_t$ and $\sigma_t$ are $\mathcal{F}_t^{W} \otimes \mathcal{B}([0,t])$-progressive processes satisfying
 \[
  \E\biggl[e^{2 \int_0^T \vert \mu_t \vert \, dt} + e^{\int_0^T \sigma_t^2 \, dt}\biggr] < \infty \qquad \mbox{and} \qquad \sigma > 0 \quad P \otimes dt\mbox{-a.e.}
 \]
 In particular, we do not assume any Markovianity of the drift or diffusion coefficient. Moreover, let the money market account be given by
 \[
 dB_t = r_t B_t \, dt, \qquad B_0 =1
 \]
 for some progressive interest process $r$ satisfying $\E \bigl[e^{\int_0^T \vert r_t \vert \, dt }\bigr]<\infty$. Define the market price of risk $\theta$ through
 \[
 \theta_t \sigma_t = \mu_t - r_t.
 \]
 To preclude arbitrage in the sense of a {\it 'free lunch with vanishing risk'}, we must assume that the market price of risk satisfies
 \[
 \E\Biggr[\mathcal{E}\biggl(-\int_0^{\bf \cdot} \theta_t \, dW^1_t \biggr)_T \Biggr] = 1,
 \]
 where $\mathcal{E}(X)_T := \exp{\bigl(X_T-1/2\langle X\rangle_T\bigr)}$ denotes the stochastic (Dol\'{e}ans-Dade) exponential of the semimartingale $X$. Additionally, for our results, we have to assume a little bit more regularity in terms of Malliavin differentiability (for a reference on Malliavin calculus see \cite{Nua}, \cite{ENua}).  We are in a one-dimensional stochastic volatility model.  Hence, the underlying Hilbert space $\mathcal{H}$ is given by $L^2([0,T];\mathbb{R})$, endowed with the canonical inner product. For $p\ge 1$ we denote by
 \begin{align}
 \mathbb{D}^{1,p}& := \biggl\{ F \in L^p(\Omega, \mathcal{F}, P) \, : \, \Vert F \Vert_{1,p} := \Bigl(\E[\vert F \vert^p]  + \E\bigl[  \Vert DF \Vert_\mathcal{H}^p \bigr]\Bigr)^\frac{1}{p} <\infty \biggr\}\label{eq:D1p} \\
  &= \Biggl\{ F \in L^p(\Omega, \mathcal{F}, P) \, : \, \Vert F \Vert_{1,p} := \biggl(\E[\vert F \vert^p]  + \E\biggl[  \int_0^T (D_tF)^p\, dt \biggr]\biggr)^\frac{1}{p} <\infty \Biggr\}\nonumber
 \end{align}
 the subspace of $L^p$ of random variables with $p$-integrable Malliavin derivatives. We note that $D_tF$ denotes the the Malliavin derivative. Moreover, denote by $\mathbb{L}^{1,2}$ the class of all processes $u \in L^2(\Omega \times [0,T])$ such that $u_t \in \mathbb{D}^{1,2}$ for almost all $t$ such that there exists a measurable version of the two-parameter process $D_s u_t$ satisfying
 \[
 \E \int_0^T\int_0^T (D_su_t)^2 \, ds \, dt < \infty.
 \]
 \begin{assumption}\label{Ass2}
 We assume that $\theta^2 \in \mathbb{L}^{1,2}$ and 
\begin{equation}
 \E\Biggr[\Biggl(\mathcal{E}\biggl(-\int_0^{\bf \cdot} \theta_t \, dW^1_t \biggr)_T\Biggr)^2 \Biggr] <\infty.
 \label{eq:Ass2}
 \end{equation}
 \end{assumption}
We note, in particular, that this Assumption is satisfied by {\em local volatility models}
\begin{equation}
dS_t = \mu( t, S_t) S_t \, dt + \sigma (t, S_t) S_t \, dW_t, \qquad S_0 =s,
\label{eq:S1}
\end{equation}
as long as $\mu(t,s)$ and $\sigma(t,s)$ are nice enough. Specifically, a sufficient condition is that $\mu(t,s)$, $\mu_s(t,s)s$, $\sigma(t,s)$, and  $\sigma_s(t,s)s$ are bounded functions, uniformly continuous in the first component and twice continuously differentiable in the second component with bounded derivatives and that $\sigma$ is uniformly bounded away from zero.

Recall that $\mathcal{M}^e$ denotes the set of all equivalent local martingale measures, in our current setting given by
\begin{equation}
\mathcal{M}^e=\bigl\{ Q \, : \,  Q \sim P, \, B^{-1}S \mbox{ is a local $\Q$-martingale} \bigr\}.
\label{eq:Me}
\end{equation}
Since we are in a complete market case, it is well-known that the set of all equivalent local martingale measures consists of a single measure
\begin{equation}\label{ELMMchar}
 \mathcal{M}^e= \Biggl\{ \Q  \, :  \, \left. \frac{d\Q}{d\P} \right\vert_{\mathcal{F}_T^{W}} = Z_T := \mathcal{E}\Bigl(-\int_0^{\bf \cdot} \theta_t \, dW^1_t \Bigr)_T\Biggr\}.
\end{equation}
 
\begin{lemma}\label{stochvoldensities}
If $\int_0^T\theta_t^2 \, dt > 0$ $\P$-a.e., then $Z_T$ has continuous law.
\end{lemma}
 
\begin{proof} 
We will show more then required -- asserting that under the stated conditions the random variable $Z_T$ has a density with respect to the Lebesgue measure. This will be done by using a Malliavin calculus-based result, which is due to Bouleau and Hirsch. For the logarithm $L_T :=\log{ Z_t}$ we have
 \[
 L_T = -\int_0^T \theta_s \, dW_s - \frac{1}{2}\int_0^T \theta_s^2\, ds. 
 \]
 It follows that
 \begin{align*}
 \E\Bigl[ \Vert D L_T\Vert_\mathcal{H}^2 \Bigr] =& \E\biggl[\int_0^T\biggl(\theta_t + \int_t^T D_t \theta_s \, dW_s + \frac{1}{2} \int_t^T D_t \theta_s^2  \, ds\biggr)^2 \, dt \biggr]\\
 \leq & \E\biggl[\int_0^T 3\biggl(\theta_t^2 + \Bigl(\int_t^T D_t \theta_s \, dW_s\Bigr)^2 + \Bigl(\frac{1}{2}\int_t^T D_t \theta_s^2 \, ds\Bigr)^2\biggr) \, dt \biggr]\\
 = & 3\E\biggl[\int_0^T \theta_t^2\, dt \biggr] +3 \E\biggl[\int_0^T\int_t^T (D_t \theta_s)^2 \, ds \, dt \biggr] +\frac{3}{4} \E\biggl[\int_0^T\biggl(\int_t^T D_t \theta_s^2  \, ds\biggr)^2 \, dt \biggr]\\
 \leq & 3 \int_0^T \E[\theta_t^2]\, dt  + 3\int_0^T \E\Bigl[ \Vert D\theta_s\Vert_\mathcal{H}^2 \Bigr]\, ds + \frac{3}{4}T^2\int_0^T \E\Bigl[ \Vert D \theta_s^2\Vert_\mathcal{H}^2 \Bigr] \, ds\\
 \leq & 3\int_0^T \Vert \theta_s \Vert_{1,2}^2 \, ds + \frac{3}{4}T^2 \int_0^T \Vert \theta_s^2 \Vert_{1,2}^2 \, ds \leq 3 \int_0^T \biggl(\frac{1}{T}+ \frac{T}{2}\Vert \theta_s^2 \Vert_{1,2}\biggr)^2 \, ds < \infty
 \end{align*}
 by Assumption \ref{Ass2}. Hence,
 \begin{align*}
  \Vert Z_T \Vert_{1,1} &= \E\bigl[Z_T \bigr] + \E\Bigl[ \Vert D Z_T\Vert_\mathcal{H} \Bigr] = \E\bigl[ Z_T  \bigr] + \E\Bigl[Z_T  \Vert D L_T\Vert_\mathcal{H} \Bigr] 
  = \E\Bigl[[Z_T  \bigl(1+ \Vert D L_T\Vert_\mathcal{H} \bigr)\Bigr] \\
  &\le \sqrt{ \E \bigl[Z_T^2\bigr]} \sqrt{\E \Bigl[  \bigl(1+\Vert D L_T\Vert_\mathcal{H}\bigr) ^2 \Bigr]    } 
 \le  \sqrt{ \E\bigl[Z_T^2\bigr]} \sqrt{2\E\Bigl[ 1+\Vert D L_T\Vert_\mathcal{H}^2  \Bigr]}<\infty.
 \end{align*}
 Following the criterium for absolute continuity (cf. \cite[Theorem 2.1.3]{Nua}), it is therefore enough to show that
 \[
 \Vert D Z_T\Vert_\mathcal{H} >0 \qquad \P\mbox{-a.s.} 
 \]
 From 
 \[
 (D Z_T) ^2 = Z_T^2 ( D L_T)^2 
 \]
 and from the fact that $Z_T>0$  $\P$-a.s., this is equivalent to the fact that $\Vert D L_T\Vert_\mathcal{H}   >0~ \P$-a.s. However, for every adapted process $Y  \in  \dom{(\delta)} \subseteq L^2(\Omega; \mathcal{H})$, the domain of the Skorohod integral, we have, by the definition of $L_T$
 \begin{align*}
 \E\Bigl[ \bigl\langle Y_t \, , \,  D_tL_T \bigr\rangle_\mathcal{H} \Bigr] &= \E\Bigl[ L_T \, \delta(Y_t)\Bigr] = \E\biggl[ L_T \biggl( \int_0^T Y_t^2\, dW_t  \biggr)\biggr] = \E\biggl[ \Bigl\langle -\int_0^\cdot \theta_t\, dW_t,\int_0^\cdot Y_t\, dW_t \Bigr\rangle_T \biggr]\\&= \E\biggl[ \Bigl\langle -\int_0^\cdot \theta_t\, dW_t,\int_0^\cdot Y_t\, dW_t \Bigr\rangle_T \biggr]
 = \E\biggl[- \int_0^T \theta_t Y_t  \, dt\biggr]. 
 \end{align*}
 Thus, we conclude that $\Vert D L_T\Vert_\mathcal{H} =0 ~ \P$-a.s., if only if $\int_0^T\theta_t^2 \, dt = 0~ \P$-a.s.\qed
 \end{proof}
\end{example}

We turn our attention now to incomplete market models:
 \begin{example}\label{ex3}
 {\bf (Lognormal mixture models)}: Similar to Example \ref{ex1} let $W$ be a one-dimensional Brownian motion defined on some probability space $(\Omega, \mathcal{F}, \P)$ and denote by $\bigl(\mathcal{F}_t^{W}\bigr)$ the filtration generated by it, augmented by all $\P$-negligible sets. Let $N\in\N$ and consider a random variable $X$ on $\{1, \ldots N\}$ such that $\nu \bigl[ X=i \bigr] = p_i,~i=1,\ldots , N$, for some counting measure $\nu$ where $p_i>0$ and $\sum_{i=1}^N p_i=1$. Let the stock price process modeled on the space $(\tilde{\Omega}, \tilde{\mathcal{F}}, \tilde{\P})=(\Omega \times \{1, \ldots,N\}, \mathcal{F} \otimes \mathcal{P}(\{1,\ldots,N\}), \P \otimes \nu)$ by
\begin{equation}
dS_t = \mu_t S_t \, dt + \sigma_t(X) S_t \, dW_t, \qquad S_0 =s,
\label{eq:S2}
\end{equation}
 where $\mu_t$ and $\sigma_t(i),~i=1,\ldots, N$ are deterministic functions, bounded and bounded away from zero, and satisfying $\sigma_t(i)=\sigma_0$ for $t\in[0,t_0]$ for some $t_0>0$, and $\sigma_t(i)\ne\sigma_t(j),~i\ne j $ for all $t\in (t_0,t_1)$, for some $t_1>t_0$, satisfying $t_1\le T$. We note that the filtration $\mathcal{F}_t^S$  generated by the stock price is not right continuous at $t_0$, and it agrees with the filtration generated by $X$ and $\mathcal{F}^W_t$ only for $t>t_0$, whereas it is strictly smaller at $t\le t_0$. Following Brigo and Mercurio \cite[Section 10.4]{Brigo} this SDE has a unique strong solution.

Let the money market account be given by
 \[
 dB_t = r_t B_t \, dt, \qquad B_0 =1,
 \]
 for some bounded progressive interest process $r$ and, conditioned on $X=i$, define the market price of risk $\theta$ through
 \[
  \theta_i  = \frac{\mu_t - r_t}{\sigma_t(i)}.
 \]
Assume also that the $\left\{\theta_i\right\}_{i=1}^n$ are linearly independent.

\begin{lemma}\label{stochvoldensities1}
The the set of all equivalent local martingale measures $\mathcal{M}^e$ is given by 
 \begin{equation}\label{ELMMchar1}
\mathcal{M}^e= \Biggl\{ \tilde{\Q}  \, :  \, \left. \frac{d\tilde{\Q}}{d\tilde{\P}} \right\vert_{\mathcal{F}_T^{S}} = Z_T^{{\q}} := \sum_{i=1}^N\frac{q_i}{p_i}\mathcal{E}\Bigl(-\int_0^{\bf \cdot} \theta_i \, dW_t \Bigr)_T,~\q \in \mathcal{K} \Biggr\},
\end{equation}
where $\mathcal{K} \define\left\{\q\in \mathbb{R}^N_{>0}\colon ~\sum_{i=1}^Nq_i=1\right\}$. If $\int_0^T\theta_i(t)^2 \, dt > 0$ $\P$-a.e. for every $i=1,\ldots, N$, then the set $\bigl\{Z_T^{\q} \, : \, \q\in\mathcal{K}\bigr\}$ has uniformly absolutely continuous distributions.
\end{lemma}
 
\begin{proof}
Note first that from \cite[Theorem 4.1]{Amen} it follows that any integrable random variable $Z_T =  \left. \frac{d\tilde{\Q}}{d\tilde{\P}} \right\vert_{\mathcal{F}_T^{S}}$ in the filtration $\mathcal{F}^S$ has the representation
\begin{equation}\label{amenrep}
Z_T = Z_{t_0+} + \int_{t_0}^T \eta_s \, dW_s
\end{equation}
where $Z_{t_0+} \in \mathcal{F}_{t_0+}^S$ and $\eta_s$ an $\mathcal{F}^S$-predictable and locally integrable process. More precisely, setting $\mathcal{F}'_t = \mathcal{F}_{t +
t_0 + \varepsilon}$ and applying \cite[Theorem 4.1]{Amen} yields $Z_T = Z_{t_0+ \varepsilon} + \int_{t_0+ \varepsilon}^ T \eta_s \, dW_s$ and sending $\varepsilon$ to zero this converges to \eqref{amenrep} as the stochastic integrals are consistently constructed and $Z_{t_0+ \varepsilon} = \E[ Z_T\, \vert \,\mathcal{F}^S_{t_0+\varepsilon}]$ is a backward martingale that converges by the backward martingale convergence theorem almost surely to $Z_{t_0+} := \E[ Z_T \, \vert \,\mathcal{F}^S_{t_0+}]$.
Moreover, using the classical martingale representation theorem this yields
\begin{align*}
Z_T & = \Delta Z_{t_0} + Z_{t_0} + \int_{t_0}^T \eta_s \, dW_s =  \Delta Z_{t_0} +  Z_0 + \int_0^{t_0} \eta_s \, dW_s+ \int_{t_0}^T \eta_s \, dW_s  \\
&= 1 + \Delta Z_{t_0} + \int_0^T \eta_s \, dW_s
\end{align*}
with $\eta_s$ an $\mathcal{F}^S$-predictable and locally integrable process and $\Delta Z_{t_0} = Z_{t_0+}-Z_{t_0}$.
Applying It\^{o}'s formula to $\log{Z_t}$ we conclude that
\begin{equation*}
 Z_t =\left( 1+  \frac{\Delta Z_{t_0}}{Z_{t_0}}\mathbbone_{\{t>t_0\}}(s)\right)  \exp{\Biggl({\int_{0}^t\frac{\eta_s}{Z_s} \, dW_s -\frac{1}{2}\int_{0}^t \biggl(\frac{\eta_s}{Z_s}\biggr)^2 \, ds\Biggr)}}
\end{equation*}
The no arbitrage condition requires that the discounted stock price is a local martingale under $\tilde{\Q}$ or, equivalently, that $B_t^{-1}S_tZ_t$ is a local $\tilde{\P}$-martingale. Noting that
\begin{align*}
B_t^{-1}S_t & =  S_0 + \int_0^t \sigma_s(i) B_s^{-1}S_s \, dW_s +  \int_0^t (\mu_s - r_s) B_s^{-1}S_s \, dt \\
&= S_0 \exp{\biggl( \int_0^t \sigma_s(i) \,  dW_t +  \int_0^t \Bigl(\mu_s - r_s - \frac{1}{2} \sigma_s^2(i)\Bigr)\, ds\biggr)}
\end{align*}
implies
\begin{align}\label{mart}
B_t^{-1}S_t Z_t &=  S_0 \cdot \left( 1+  \frac{\Delta Z_{t_0}}{Z_{t_0}}\mathbbone_{\{t>t_0\}}(s)\right) \nonumber\\
& \phantom{=} \cdot \exp{\Biggl( \int_0^t \Bigl(\sigma_s(i) + \frac{\eta_s}{Z_s}\Bigr)\,  dW_t +  \int_0^t \Bigl(\mu_s - r_s - \frac{1}{2} \sigma_s^2(i) - \frac{1}{2}\biggr(\frac{\eta_s}{Z_s}\biggr)^2 \, ds\Biggr)}
\end{align}
we have just to check under which conditions \eqref{mart} is a $\tilde{\mathbb{P}}$-martingale. Note that it is enough to check when the exponential is a local martingale (checking that for $\E[B_t^{-1}S_t Z_t  \, \vert \mathcal{F}^S_s]$ for $t\geq s>t_0$ and $t_0 \geq t \geq s$ is straightforward, at $t_0$ this follows from the continuity of the exponential). Calculating the quadratic variation of the stochastic integral and comparing it with the determinist one yields
\[
\frac{\eta_s}{Z_s} = \frac{r_s-\mu_s}{\sigma_s(i)} = \theta_i(s).
\]
Moreover, as $1+  \frac{\Delta Z_{t_0}}{Z_{t_0}} = \frac{Z_{t_0+}}{Z_{t_0}} \in \sigma(X)$ is only supported on $\{1, \ldots N\}$ and the martingale condition forces $\E\bigl[1+  \frac{\Delta Z_{t_0}}{Z_{t_0}} \, \big\vert \mathcal{F}_{t_0}^S\bigr]=1$, we get that the representation
\eqref{ELMMchar1} is a necessary condition on equivalent local martingale measures. It is straightforward to check that it is also sufficient.

Additionally, we have that all $\theta_i,~i=1,\ldots, N$ are uniformly bounded.  It follows that the densities of the normally distributed random variables $\int_0^{T} \theta_i \, dW_t$ are also uniformly bounded. Furthermore, the Gramian matrix $G(\theta_1,\ldots, \theta_N)$ with elements $G_{ij}= \int_0^{T} \theta_i \theta_j\, dt,~1\le i,j\le N$, has non-zero Gram determinant, by our assumption that $\left\{\theta_i\right\}_{i=1}^N$ are linearly independent. It follows that the random vector
\[
{\overrightarrow \nu}\define\left( \int_0^{T} \theta_1 \, dW_t,\ldots , \int_0^{T} \theta_N \, dW_t\right)
\]
is normally distributed with mean zero, and variance $G(\theta_1,\ldots, \theta_N)$. Hence, it has a bounded density. Finally, note that for $\q\in\mathcal{K}$ the random variable
\[
Z_T^{{\q}} := \sum_{i=1}^N\frac{q_i}{p_i}\mathcal{E}\Bigl(-\int_0^{\bf \cdot} \theta_i \, dW_t \Bigr)_T
\]
is an exponential transformation of ${\overrightarrow \nu}$, since the functions $\theta_i,~i=1,\ldots, N$ are
deterministic. Therefore, it has a density, too. This transformation is continuous as a function of $\q$ and, thus, so is the density of $Z_T^{{\q}}$. It follows that the density of $Z_T^{{\q}}$ is uniformly bounded for any $\q$ in the compact closure $\overline {\mathcal{K}}$.\qed
\end{proof}

As the family $\bigl\{Z_T^{\q} \, : \, \q\in \mathbb{R}^N_{>0}, ~\sum_{i=1}^Nq_i=1 \bigr\}$ has uniformly absolutely continuous distributions, we can then apply Proposition \ref{suffCrit} to establish the distributional continuity of the dual optimizer.
\end{example}

\subsection*{Stochastic Volatility Examples}\label{ex:stochVol}

We now generalize the setting of Example \ref{ex1} to encompass stochastic volatility models, dropping the assumption of market completeness. Namely, let $W^1$ and $W^2$ be two independent one-dimensional Brownian motions (the generalization to the  multi-dimensional case is again straightforward) defined on some probability space $(\Omega, \mathcal{F}, P)$ and denote by $\bigl(\mathcal{F}_t^{W^1,W^2}\bigr)$ the filtration generated by them, augmented by all $\P$-negligible sets (as usual we assume without loss of generality that $\mathcal{F}_T^{W^1,W^2}=\mathcal{F}$). Let the stock price process given by
\begin{equation*}
dS_t = \mu_t S_t \, dt + \sigma_t S_t \, dW^1_t, \qquad S_0 =s,
\end{equation*}
where $\mu_t$ and $\sigma_t$ are $\mathcal{F}_t^{W^1,W^2} \otimes \mathcal{B}([0,t])$-progressive processes satisfying
 \[
  \E\biggl[e^{2 \int_0^T \vert \mu_t \vert \, dt} + e^{\int_0^T \sigma_t^2 \, dt}\biggr] < \infty \qquad \mbox{and} \qquad \sigma > 0 \quad P \otimes dt\mbox{-a.e.}
 \]
 In particular we again do not assume any Markovianity of the drift or diffusion coefficient. Moreover, we still assume that the money market account be given by
\[
dB_t = r_t B_t \, dt, \qquad B_0 =1
\]
for some progressive interest process $r$ satisfying $\E \bigl[e^{\int_0^T \vert r_t \vert \, dt }\bigr]<\infty$ and define the market price of risk $\theta$ through
\[
\theta_t \sigma_t = \mu_t - r_t.
\]
We readily adjust all of the remaining definitions of Example \ref{ex1} to this framework. We note that, in this case, $\mathcal{H}$ will be given by $L^2([0,T];\mathbb{R}^2)$. The definition of $\mathbb{D}^{1,p}$ in \eqref{eq:D1p} will not change. Finally, we will adjust the definition $\mathbb{L}^{1,2}$ to be the class of all processes $u \in L^2(\Omega \times [0,T])$ such that $u_t \in \mathbb{D}^{1,2}$ for almost all $t$ such that there exists a measurable version of the two-parameter process $D_su_t$ satisfying
\[
\E \int_0^T\int_0^T (D_s^1u_t)^2 + (D_s^2u_t)^2 \, ds \, dt < \infty.
\]
We, of course, still preclude arbitrage in the sense of a {\it 'free lunch with vanishing risk'} by assuming that the market price of risk satisfies
\[
\E\Biggr[\mathcal{E}\biggl(-\int_0^{\bf \cdot} \theta_t \, dW^1_t \biggr)_T \Biggr] = 1,
\]
and we still assume Assumption \ref{Ass2} is satisfied. 

Recall that $\mathcal{M}^e$, given by \eqref{eq:Me}, denotes the set of all equivalent local martingale measures. The first major change, in comparison to Example \ref{ex1}, is that, in our current setting, $\mathcal{M}^e$ is given by
\begin{lemma}\label{stochvoldensities2}
The set of all equivalent local martingale measures can be characterized as
\begin{align}\label{ELMMchar2}
\mathcal{M}^e & = \Biggl\{ \Q  \, :  \, \left. \frac{d\Q}{d\P} \right\vert_{\mathcal{F}_T^{W^1,W^2}} = Z_T^\theta(\lambda) := \mathcal{E}\Bigl(-\int_0^{\bf \cdot} \theta_t \, dW^1_t + \int_0^{\bf \cdot} \lambda_t \, dW^2_t\Bigr)_T, \Biggr. \nonumber \\
& \phantom{= \Biggl\{\, } \Biggl.\lambda\in \Lambda, \, \E[Z_T^\theta(\lambda)]=1\Biggr\},
\end{align}
where
\[ 
\Lambda :=  \Bigl\{ \lambda \mbox{ predictable, such that } \int_0^T \lambda_t^2 \, dt < \infty \, \P\mbox{-a.s.}\Bigr\}.
\]
Moreover, if $\int_0^T\theta_t^2 \, dt > 0$ $\P$-a.e., then the random variable $Z_T^\theta(0)$ has a continuous law.
\end{lemma}

\begin{proof}
To begin, we prove the characterization of the set of equivalent local martingale measures, which follows the proof for the classical Markovian case (cf. \cite{Frey}). First, it is clear that, only under the condition $\E[Z_T^\theta(\lambda)]=1$, will the new measure $\Q$ be a probability measure. By the martingale representation theorem, we know that we can find predictable processes $\eta$, $\xi$ such that
\[
Z_t^\theta(\lambda) =\left.  \frac{d\Q}{d\P} \right\vert_{\mathcal{F}_t^{W^1,W^2}} = 1+ \int_0^t \eta_s \, dW^1_s + \int_0^t \xi_s \, dW^2_s.
 \]
Since $\Q$ and $\P$ are equivalent, the density process is strictly positive and we can define its logarithm $L_t^\theta(\lambda) = \log{ Z_t^\theta(\lambda)}$ which satisfies, by It\^{o}'s formula
\[
dL_t^\theta(\lambda) =   \frac{\eta_t}{Z_t^\theta(\lambda)} \, dW^1_t + \frac{\xi_t}{Z_t^\theta(\lambda)} \, dW^2_t - \frac{1}{2}  \Biggl(\biggl(\frac{\eta_t}{Z_t^\theta(\lambda)}\biggr)^2+ \biggl(\frac{\xi_t}{Z_t^\theta(\lambda)}\biggr)^2 \Biggr)\, dt.
\]
Expressing now the stock-price process $S$ under $\Q$, we get, by Girsanov's theorem
\begin{align}\label{stockchange}
dB_t^{-1}S_t & = B_t^{-1}\Bigl( \sigma_tS_t dW^1_t + (\mu_t - r_t) S_t \, dt\Bigr) \nonumber \\
& = B_t^{-1}\Biggl( \sigma_t S_t dW_t^\Q +  S_t\biggl( \mu_t -r_t+ \sigma_t\frac{\eta_t}{Z_s^\theta(\lambda)}\biggr)\, dt\Biggr)
\end{align}
for some $\Q$-Brownian motion $W^\Q$ independent of $W^2$. Hence, the discounted stock price is a local martingale only if
\[
\frac{\eta}{Z^\theta(\lambda)} = - \theta \qquad \mbox{and} \qquad \frac{\xi}{Z^\theta(\lambda)} = \lambda
\]
for some predictable, square-integrable process $\lambda$. On the other hand, every expression on the right hand side of \eqref{ELMMchar2} defines an equivalent probability measure. By \eqref{stockchange}, the stock price is a local martingale under this measure.

The second part of the assertion -- that $Z_T^\theta(0) $ has a continuous law -- is a direct consequence of Lemma \ref{stochvoldensities}.\qed
\end{proof}

Next, we would like to apply Proposition \ref{suffCrit} to conclude that $\hat Y_T(y)$ has a continuous law. But, this is not easy to do, since we need to satisfy the assumption of Proposition \ref{suffCrit} that the family $Z_T^\theta(\lambda)$ is uniformly absolutely continuous with respect to the Lebesgue measure. Here, we propose an alternative approach, which requires a certain restriction, namely that the market price of risk does not depend on the Brownian motion driving the stochastic volatility. The classical example of this would be a constant market price of risk.  While this is a shortcoming from a theoretical point of view, it does not matter much from a practical perspective. The straight-forward idea of using a constant drift does not work out nicely in many situations, as there may not exist an equivalent change of measure (e.g., in the Stein \& Stein model or in the Heston model without Feller condition - c.f. \cite{HW}). A standard 
procedure (compare the discussion in \cite[Section 2.4.2]{FPSS}) is exactly to assume a constant market price of risk to avoid these integrability issues.
\begin{lemma}\label{lemma:dual-minimizer}
Assume that the market price of risk $\theta_t\in\mathcal{F}_t^{W^1}$. Then the infimum over all equivalent local martingale measures in the calculation of the value function of the dual problem $v$ is reached for $\lambda=0$, i.e.,
\[
v(y) = \inf_{\Q \in \mathcal{M}^e} \E\biggl[ \bar{U}^*\biggl(y \frac{d\Q}{d\P}\biggr)\biggr]=\E\biggl[ \bar{U}^*\biggl(y Z_T^\theta(0)\biggr)\biggr] .
\]
\end{lemma}

\begin{proof}
From Jensen's inequality it follows that for any $\lambda\in\Lambda$
\begin{align*}
&\E\Bigl[ \bar{U}^*\bigl(y Z_T^\theta(\lambda)\bigr)\Bigr] = \E\Biggl[ \left. \E\biggl[ \bar{U}^*\biggl(y \mathcal{E}\Bigl(-\int_0^{\bf \cdot} \theta_t \, dW^1_t \Bigr)_T\mathcal{E}\Bigl( \int_0^{\bf \cdot} \lambda_t \, dW^2_t\Bigr)_T\biggr) \, \right\vert \, \mathcal{F}_{\bf \cdot}^{W^1}\biggr] \Biggr]\\
&\ge \E\Biggl[ \bar{U}^*\Biggl(y \mathcal{E}\Bigl(-\int_0^{\bf \cdot} \theta_t \, dW^1_t \Bigr)_T\E\biggl[\mathcal{E}\Bigr( \int_0^{\bf \cdot} \lambda_t \, dW^2_t\Bigr)_T\Big|\mathcal{F}_{\bf \cdot}^{W^1} \biggr]\Biggr)\Biggr].
\end{align*}
Since $W^1$ and $W^2$ are independent, $\E\bigl[\mathcal{E}\bigl( \int_0^{\bf \cdot} \lambda_t \, dW^2_t\bigr)_T\big|\mathcal{F}_{\bf \cdot}^{W^1} \bigr]=1$ a.s., and we note that 
\begin{align*}
&\E\Bigl[ \bar{U}^*\bigl(y Z_T^\theta(\lambda)\bigr)\Bigr]  \ge \E\biggl[ \bar{U}^*\biggl(y \mathcal{E}\Bigl(-\int_0^{\bf \cdot} \theta_t \, dW^1_t \Bigr)_T\biggr)\biggr] = \E\Bigl[ \bar{U}^*\bigl(y Z_T^\theta(0)\bigr)\Bigr].
\end{align*}
Since, of course, $v(y) \le  \E\bigl[ \bar{U}^*\bigl(y Z_T^\theta(0)\bigr)\Bigr]$, we conclude that 
\begin{align*}
v(y) &=  \E\Bigl[ \bar{U}^*\bigl(y Z_T^\theta(0)\bigr)\Bigr].
\end{align*}
\qed
\end{proof}

Thus, we can use Lemma \ref{stochvoldensities} and apply Theorem \ref{general} directly:

\begin{theorem}\label{theoSV}
Assume that for a stochastic volatility model Assumption \ref{Ass2} holds. Additionally, assume that $\theta_t \in \mathcal{F}_t^{W_1}$, and $\int_0^T\theta_t^2 \, dt > 0~ \P$-a.s. Then the original problem \eqref{original} has a maximizer, which is also a maximizer of the concavified problem \eqref{concavified}.
\end{theorem}

Finally, we want to show that the Assumption \ref{Ass2} is satisfied in many standard volatility models. First we remark that if the volatility function $\sigma_t = \sigma(Y_t)$ is a smooth function in $Y_t$, bounded and bounded away from zero, and the volatility process satisfies $Y \in \mathbb{L}^{1,2}$, then the assumption is satisfied. This is also enough to ensure the existence of an equivalent local martingale measure. Turning to more standard models, we observe that it can be shown that many standard volatility processes, such as e.g., Ornstein-Uhlenbeck, CIR or geometric Brownian motion, satisfy the Malliavin differentiability condition (for the CIR process, at least in the nice regime when Feller's condition holds -- cf. \cite{AlosEwald}).  Thus, all possible problems arise, not from the Malliavin smoothness condition, but from the requirement that $\theta_t\in\mathcal{F}_t^{W^1}$, which is usually not satisfied for constant drift. As mentioned above, the standard way 
to circumvent this problem is to allow for a volatility-dependent excess appreciation:

\begin{example}
{\bf (Correlated Hull-White model)}: We consider a bond with constant interest rate $r$ and the stock price given by
\begin{align*}
dS_t &= \Bigl( r + Y_tf(W_t^1)\Bigr) S_t \, dt + Y_t S_t \,dW_t^1, \quad S_0 =s,\\
dY_t &= b Y_t \, dt + \rho a Y_t \, dW_t^1 + \sqrt{1-\rho^2} a Y_t \, dW_t^2, \quad Y_0 = y,
\end{align*}
for constants, $b \in \mathbb{R}$, $a$, $s$, $y>0$ and $\rho \in (-1,1)$ and independent Brownian motions $W^1$, $W^2$. Moreover we assume that the excess appreciation rate $f(W_t^1)$ is given via a bounded $C^1(\Rpplus)$-function $f$ with bounded derivative, and that it is not identically zero. This guarantees that the market price of risk $\theta_t = f(W_t^1)$ remains bounded and ensures that the integrability condition of Assumption \ref{Ass2} is satisfied.

Calculating the Malliavin derivative of the $\theta^2$
\[
D_t\theta_s^2 =  D_t \bigl(f(W_s^1)^2\bigr) =  2  f'(W_s^1)  D_t W^1_s \mathbbone_{[t,T]}(s) = \begin{pmatrix} 2  f'(W_s^1) \mathbbone_{[t,T]}(s)\\ 0\end{pmatrix}.
\]
We can conclude that $\theta^2 \in \mathbb{L}^{1,2}$ since $f$ and $f'$ are bounded.
Moreover, $\int_0^T\theta_t^2 \, dt > 0~ \P$-a.s since $f$ is continuous and not identically zero. Thus, all the conditions of Lemma \ref{stochvoldensities} are satisfied.
\end{example}

The proofs for the additional two examples follow the proof of the previous example, and are thus omitted.
\begin{example}
{\bf (Correlated Scott model)}: The Scott (or exponential Ornstein-Uhlen\-beck) model is given (besides the bond with constant interest rate $r$) by the stock price dynamics
\begin{align*}
dS_t &= \Bigl( r + f(W_t^1)e^{Y_t}\Bigr)  S_t \, dt + e^{Y_t} S_t \,dW_t^1, \quad S_0 =s,\\
dY_t &= \kappa (\theta - Y_t) \, dt + \rho \xi \, dW_t^1 + \sqrt{1-\rho^2} \xi \, dW_t^2, \quad Y_0 = y,
\end{align*}
for constants,  $\kappa$, $\theta$, $\xi$, $s$, $y>0$ and $\rho \in (-1,1)$ and independent Brownian motions $W^1$, $W^2$. Again we assume that the excess appreciation rate $f(W_t^1)$ is given via a bounded $C^1(\R)$-function $f$ with bounded derivative, and that it is not identically zero.  This guarantees that the market price of risk $\theta_t = f(W_t^1)$ remains bounded and ensures the integrability condition of Assumption \ref{Ass2}.
\end{example}

\begin{example}
{\bf (Correlated Heston model under Feller condition)}:  We consider a bond with constant interest rate $r$ and the stock price given by
\begin{align*}
dS_t &= \bigl(r + f(W_t^1)\sqrt{Y_t}\bigr) S_t \, dt + \sqrt{Y_t} S_t \,dW_t^1, \quad S_0 =s,\\
dY_t &= \kappa\bigl( \theta- Y_t \bigr)\, dt + \rho \xi \sqrt{Y_t} \, dW_t^1 + \sqrt{1-\rho^2} \xi \sqrt{Y_t} \, dW_t^2, \quad Y_0 = y,
\end{align*}
for constants $\kappa$, $\theta$, $\xi$, $s$, $y>0$ and $\rho \in (-1,1)$ and independent Brownian motions $W^1$, $W^2$. Moreover, we impose the Feller condition $2 \kappa \theta > \xi^2$. Again, we assume that the excess appreciation rate $f(W_t^1)$ is given via a bounded $C^1(\R)$-function $f$ with bounded derivative, and that it is not identically zero.  This guarantees that the market price of risk $\theta_t = f(W_t^1)$ remains bounded and ensures the that integrability condition of Assumption \ref{Ass2} is satisfied.
\end{example}

Finally, we would like to remark that the same reasoning applied above to stochastic volatility models also holds true for Markovian regime switching models, as they have the same kind of representation of equivalent local martingale measures (c.f. \cite{Siu}). As long as the market price of risk is positive, sufficiently (Skorohod-) inte\-grable and depends only on the stock-driving Brownian motion, the infimum is attained independently of the volatility risk. Therefore, we conclude that the optimizer has a density.

\section{Preliminaries - The Classical Utility Optimization Problem}\label{sec:prelim}

We will now briefly review the classical results of utility optimization. We adapt statements of \cite{BTZ} and \cite{WZ} (mainly Theorem 3.2 of \cite{BTZ}) on non-smooth utility maximization for the use in our setting. To keep notation concise and well-integrated with the rest of the paper, we hide our incentive scheme by setting $g(x)=x$ throughout this section. To emphasize this we will talk about the {\em classical} utility optimization problems 
\begin{equation}\label{primal}
u(x) := \sup_{X \in \mathcal{X}(x)} \E\bigl[U\bigl(X_T\bigr)\bigr],
\end{equation}
and its dual
\begin{equation}\label{dual1}
v(y) := \inf_{Y \in \mathcal{Y}(y)} \E \Bigl[U^*\bigl(Y_T\bigr)\Bigr].
\end{equation}
We will continue with our standing Assumptions \ref{assumpt:finite} and \ref{equivlMart} throughout this section. However, we will later relax Assumption \ref{UassKS}. The central result of Kramkov and {Scha\-chermayer \cite[Theorem 2.2.]{KramSchach1} is the following:

\begin{theorem}[Kramkov-Schachermayer]\label{thm1}
Under the Assumptions \ref{assumpt:finite}, \ref{equivlMart}, and \ref{UassKS}, for the utility maximization problem \eqref{primal}, it holds that
\begin{itemize}\label{KSduality}
\item[a)]The functions $u$ and $v$ are finite on $\mathbb{R}_{> 0}$ and conjugate, i.e., $v = u^*$. Moreover $u$ and $-v$ are strictly concave, strictly increasing, continuously differentiable on $\Rplus$, satisfy the Inada conditions \eqref{Inada} and $u$ satisfies the asymptotic elasticity condition \eqref{AE}.
\item[b)]The optimal solutions $\hat{X}(x) \in \mathcal{X}(x)$ for \eqref{primal} and $\hat{Y}(y) \in \mathcal{Y}(y)$ for \eqref{dual1} exist, are unique and are for $y = u'(x)$ related through
\[
\hat{X}_T(x) = -(U^*)'\bigl(\hat{Y}_T(y)\bigr), \qquad \qquad \hat{Y}_T(y) = U'\bigl(\hat{X}_T(x)\bigr).
\]
Moreover, $\hat{X}(x)\hat{Y}(y)$ is a uniformly integrable martingale.
\item[c)]Additionally we have
\[
v(y) = \inf_{\Q \in \mathcal{M}^e} \E\biggl[ U^*\biggl(y \frac{d\Q}{d\P}\biggr)\biggr],
\]
however the infimum is in general not attained in $\mathcal{M}^e$.
\end{itemize}
\end{theorem}

Asymptotic elasticity is the minimal condition to assure the duality result in general semimartingale models for smooth utility functions (cf. \cite{KramSchach1}). (If one poses a joint condition on model and utility function, then the minimal condition is the finiteness of the dual value function, cf. \cite{KramSchach2}.) However, as previously mentioned, the concavified utility function $\bar{U}^{**}$ will be, in general, neither strictly concave nor satisfy the Inada condition at $0$. Thus, we will have to rely on results for nonsmooth utility maximization. While we will still impose Assumptions \ref{assumpt:finite} and \ref{equivlMart}, we will have to relax Assumption \ref{UassKS}. In the nonsmooth case, it turns out that the asymptotic elasticity -- following Deelstra, Pham, and Touzi \cite{DPhT} -- has to be written on the convex conjugate of the utility function. The following general result is due to Bouchard, Touzi and Zeghal \cite[Theorem 3.2.]{BTZ}. A simplification of the proof can be found in 
Westray and Zheng \cite[Theorem 5.1.]{WZ}.

We relax the conditions on the utility function $U$, assuming only that $U: (\alpha,\infty)\rightarrow \R$, $\alpha \in \R$, is nonconstant, nondecreasing and concave (we extend $U$ again continuously to $[\alpha, \infty)$, allowing the value $-\infty$ at $\alpha$ while still assuming that $U(\infty)>0$). In particular, we no longer assume that $U$ is continuously differentiable on $(\alpha,\infty)$ nor do we require that $U$ to be strictly increasing or strictly concave. Finally, we no longer impose Inada conditions, but merely that the closure of the domain of the dual function is $\R_{\ge 0}$.  As mentioned above, the asymptotic elasticity condition will be written on the dual function. Hence, we substitute the following assumption for Assumption \ref{UassKS}
\begin{assumption}\label{UassBTZ}
The investor's preferences are represented by a utility function $U: (\alpha,\infty)\rightarrow \R$.
\begin{itemize}
\item[a)] We assume that $U$ is nonconstant, nondecreasing and concave;
\item[b)] The dual function satisfies $\overline{\dom{U^{*}}} = \Rpplus$;
\item[c)] Moreover, it satisfies the dual asymptotic elasticity condition
\begin{equation}\label{AEdual}
AE^*(U) := \limsup_{y \to 0} \sup_{x \in -\partial U^*(y)} \frac{yx}{U^*(y)} < \infty. 
\end{equation}
\item[d)] There exists $y>0$ such that $v(y)$ defined by \eqref{dual1} is finite.
\end{itemize}
\end{assumption}

\begin{remark}\label{AEequiv}
 We note that for smooth $U$ the classical and dual asymptotic elasticity condition are equivalent under Inada-type conditions (cf. \cite[Proposition 4.1.]{DPhT} for a precise statement).
\end{remark}

\begin{theorem}[Bouchard-Touzi-Zeghal]\label{BTZduality}
Assume that Assumptions \ref{assumpt:finite}, \ref{equivlMart}, and \ref{UassBTZ} are satisfied, then for the optimization problems \eqref{primal} and \eqref{dual1} it holds that
\begin{itemize}
\item[a)] The functions $u$ and $v$ are finite on $(\alpha, \infty)$ and $\mathbb{R}_{> 0}$ respectively, and conjugate, i.e., $v = u^*$.
\item[b)]  Optimal solutions $\hat{X}(x) \in \mathcal{X}(x)$ for \eqref{primal} and $\hat{Y}(y) \in \mathcal{Y}(y)$ for \eqref{dual1} exist such that for some $y \in \partial u(x)$ we have that $\hat{X}(x)\hat{Y}(y)$ is a uniformly integrable martingale and
\begin{align}
\hat{X}_T(x) \in -\partial U^*\bigl(\hat{Y}_T(y)\bigr).
\label{eq:X-hat}
\end{align}
\item[c)] Additionally we have
\begin{align}
v(y) = \inf_{\Q \in \mathcal{M}^e} \E\biggl[ U^*\biggl(y \frac{d\Q}{d\P}\biggr)\biggr],
\label{eq:V-min}
\end{align}
however the infimum is in general not attained in $\mathcal{M}^e$.
\end{itemize}
\end{theorem}

Note that the subdifferential-valued random variables in part b) should be understood as random variables whose range is a subset of the image of a random variable under a set-valued function. This is a much larger set then just the collection of random variables one would obtain by picking only fixed elements in the subdifferential and looking on images under these mappings. In the first case we can have a different mapping for every $\omega \in \Omega$, whereas in the second case one fixes a single function for all $\omega$.

\begin{proof}
We adapted here the statement of \cite{BTZ} and \cite{WZ} to better fit the framework with \cite{KramSchach1}.The fact that the formulations in [3] and [29] differ from the formulations in [16] stems from the goal of the authors of [3] and [29] to accommodate a discussion of portfolio optimization on the whole real line with random initial endowment (as opposed to [16] who consider a simple portfolio optimization problem on the positive half of real line). However, their formulations (in terms of processes or terminal random variables) are equivalent for our case (without random endowment); they are (using the terminology of Kramkov/Schachermayer) the concrete and the abstract side of the same problem. The ultimate reason for the equality of both formulations is that the set of nonnegative $\mathcal{F}_T$-measurable random variables dominated by some $Y_T$, $Y \in \mathcal{Y}(y)$, is the bipolar of the set $\bigl\{y \frac{d\Q}{d\P} \, : \, \Q \in \mathcal{M}^e\bigr\}$. This is due to the bipolar theorem on 
the cone of nonnegative random variables proved by Brannath and Schachermayer \cite{BS}. For details, see \cite[Proposition 3.1 and Section 4]{KramSchach1}.

First, without loss of generality, we may assume that $\alpha=0$, otherwise, shift everything by $\alpha$. Next, we can apply Theorem 3.2 of \cite{BTZ} with $B=0,\beta=0$. It is not hard to see that, by Assumption \ref{assumpt:finite} and by the concavity of $U$, the function $u$ is finite on $\Rpplus$. The fact that $v$ is finite follows directly from Lemma 5.4 of \cite{WZ} and from Assumption \ref{UassBTZ}. One concludes that $v=u^{*}$ from Theorem 3.2 part (iii) of \cite{BTZ}.

The existence of optimal solutions $\hat X(x)$ and $\hat Y(y)$ follows from parts (i) and (ii) of Theorem 3.2 of \cite{BTZ}, respectively. The fact that $y\in\partial u(x)$ is a consequence of part (i). Additionally, \eqref{eq:X-hat} and the fact that $\hat{X}(x)\hat{Y}(y)$ is a (uniformly integrable) martingale follow from part (iii). Finally, \eqref{eq:V-min} is obtained from Remark 3.9 part 1. of \cite{BTZ}.
\qed
\end{proof}

Finally, we note that the solutions are, in general, not unique and that the value function may not be smooth. Moreover, there may well exist a random variable $Z \in -\partial U^*\bigl(\hat{Y}_T(y)\bigr)$ satisfying $\E[Z\hat{Y}_T(y)] = xy$, which is not dominated  by the terminal value of any $X(x) \in \mathcal{X}(x)$, as shown by Westray and Zheng in \cite{WZcounter}.

\section{The Dual and the Concavified Problem}\label{genresult}

Keeping in mind our standing Assumptions \ref{assumpt:finite}-\ref{UassKS}, we resume our discussion about the portfolio manager's maximization problem \ref{original}, the concavified problem \ref{concavified}, and their common dual problem \ref{dualgen}. Our plan now is to apply Theorem \ref{BTZduality}. Therefore, must first ensure that $\bar{U}^{**}$ satisfies all the conditions of Theorem \ref{BTZduality}. We also collect some properties of this function:
\begin{proposition}\label{technicallemma}
For the concavified utility function $\bar{U}^{**}$ we have
\[
\overline{\dom{\bar{U}^{**}}} = [\beta, \infty),  \qquad \beta := \inf \{ x >0 : \bar{U} (x) >-\infty\} \in [0,\infty).
\] 
Furthermore, $\bar{U}^{**}$, together with its conjugate $\bar{U}^*$, enjoys the following regularity properties. $\bar{U}^{**}$ is continuously differentiable on $(\beta, \infty)$; $\bar{U}^*$ is strictly convex on the whole domain if $U(0)=-\infty $, otherwise it is strictly convex on $\bigl(0,\bigl(\bar{U}^{**}\bigr)'(0)\bigr)$ and constant $\bar{U}(0) = \bar{U}^{**}(0)$ on $\bigl[\bigl(\bar{U}^{**}\bigr)'(0),\infty\bigr)$. Finally, $\bar{U}^{**}$ satisfies Assumption \ref{assumpt:finite} and Assumption \ref{UassBTZ}.
\end{proposition}
We divide the proof into three lemmas. The proof is elementary, but rather technical, so it can be safely skipped on the first reading.

 \begin{lemma}\label{conditions}
For the concavified utility function $\bar{U}^{**}$ it holds that $\overline{\dom{\bar{U}^{**}}} = [\beta, \infty)$, and it satisfies Assumption \ref{assumpt:finite} and Assumption \ref{UassBTZ}, a), b) and d).
\end{lemma}
\begin{proof}

Consider first the case $U(0) > -\infty$. Note that since $\bar{U}$ is continuous, its epigraph is closed and thus $\bar{U}^{**}$ is its concave hull. Thus, by Caratheodory's theorem (cf. \cite[Theorem A.1.3.6.]{HUL}), we know that
\[
 \big(0, \bar{U}^{**}(0)\bigr) \in \Biggl\{\sum_{i=1}^3 \lambda_i z_i \, : \,  \lambda_i\ge0, \, \sum_{i=1}^3\lambda_i=1, \, z_i\in \hypo \bar{U} \Biggr\}.
\]
Since $\hypo \bar{U} \subseteq \Rpplus \times \R$, it follows that the linear combination has to be the trivial, i.e, $\big(0, \bar{U}^{**}(0)\bigr) \in \hypo \bar{U}$ and $\bar{U}(0)=\bar{U}^{**}(0) >-\infty$. Thus, it follows that $\dom \bar{U}^{**} = [0, \infty) = [\beta,\infty)$.  

Similarly, if $U(0)=-\infty$, we note first that if $g(0) >0$, we have $\bar{U}(0)>-\infty$ and $\beta = 0$. Thus, we can conclude, exactly as in the previous case, that $\dom \bar{U}^{**} = [\beta,\infty)$. However, if $g(0)=0$, we know by the definition of $\beta$ that $\bar{U}(x_0)$ is real valued if and only if $x_0 \in (\beta, \infty)$. In this case, the assumption that $\bar U^{**}(\beta)>-\infty$ leads to a contradiction by Caratheodory's theorem. It follows that $\bar U^{**}(\beta)=\bar U(\beta)$ and hence $\dom \bar{U}^{**} = (\beta, \infty)$. Putting the information from all three cases together we recover the statement $\overline{\dom \bar{U}^{**}} = [\beta, \infty)$.

Now, set $b\define g(\beta)$ and note that $\bar U(x) \le U(x+b)$. We have, for $y >0$
\begin{equation}
\bar U^{*}(y) =\sup_{x > \beta} \Bigl(\bar U(x)- xy\Bigr) \le \sup_{x > \beta} \Bigl( U(x+b)- (x+b)y\Bigr) +by\le U^{*}(y)+by<\infty.
\label{eq:barU-star}
\end{equation}
Hence, we have $\overline{\mbox{ dom}(\bar U^{*})}=\Rpplus$, i.e., part b) of Assumption \ref{UassBTZ} is satisfied. It is straight forward to see that part a) holds for the concavification of a nondecreasing, nonconstant function. 
Finally, using the above, it follows also for $x>\beta$
\[
\bar U^{**}(x) =\sup_{y>0} \Bigl(\bar U^{*}(y)+xy\Bigr)\le \sup_{y>0}\Bigl(U^{*}(y)+by+xy\Bigr) = U^{**}(x+b) = U(x+b).
\]
We conclude by Theorem \ref{thm1} that
\[
w(x)\le\sup_{X\in \mathcal{X}(x)} \E\bigl[ {U}\bigl(X_T+b\bigr)\bigr]\le\sup_{X\in \mathcal{X}(x+b)} \E\bigl[ U\bigl(X_T\bigr)\bigr]
\]
 is finite on $(\beta, \infty)$. This proves Assumption \ref{assumpt:finite}.  Moreover, from \eqref{eq:barU-star}, we see that
\begin{equation}
v(y) \le \inf_{Y \in \mathcal{Y}(y)} \E\bigl[ {U}^{*}\bigl(Y_T\bigr)\bigr] + by.
\label{eq:v-inequality}
\end{equation}
The right hand side of \eqref{eq:v-inequality} is finite by Theorem \ref{thm1}. Thus, $v$ is finite on $\Rplus$, and Assumption \ref{UassBTZ} part d) is satisfied.
Hence, all the requirements of Theorem \ref{BTZduality}  are satisfied except  c) of Assumption \ref{UassBTZ}, whose proof we postpone to Lemma \ref{AElemma}.
\qed
\end{proof}

\begin{lemma}\label{Uconcsmooth}{${}$}
The concavified utility function $\bar{U}^{**}$ and its conjugate $\bar{U}^*$ enjoy the following regularity properties:
\begin{itemize}
\item[a)] The concavified utility function $\bar{U}^{**}$ is continuously differentiable on $(\beta, \infty)$. 
\item[b)] The dual utility function $\bar{U}^*$ is strictly convex on the whole domain if $U(0)=-\infty $, otherwise it is strictly convex on $\bigl(0,\bigl(\bar{U}^{**}\bigr)'(0)\bigr)$ and constant $\bar{U}(0) = \bar{U}^{**}(0)$ on $\bigl[\bigl(\bar{U}^{**}\bigr)'(0),\infty\bigr)$.
\end{itemize}
\end{lemma}
\begin{proof}
To prove a), we note first that the set  $A :=  \{ x > \beta\, : \, \bar{U}(x) \neq \bar{U}^{**}(x)\}$ where $\bar{U}$ and $\bar{U}^{**}$ do not agree, is a countable union of pairwise disjoint open intervals. We also note that the function $\bar{U}$ is continuous on $(\beta, \infty)$ since it is the composition of continuous functions ($g$ is as convex, nondecreasing function, thus continuous). The same is true for $\bar{U}^{**}$, which is a concave function by definition. Hence, $A$ is the $0$-sublevel set of the continuous function $\bar{U}^{**}-\bar{U}$ and is thus open. But every open set in $\R$ can be written as countable union of pairwise disjoint open intervals, say $A = \bigcup_{n=1}^\infty (a^-_n,a^+_n)$, $a_n^-<a_n^{+}$. We note explicitly that  $a_1^-=\beta$ and $a_n^+ = \infty $ for some $n$ are allowed. On every interval in $A$ the function $\bar{U}^{**}$ is affine (the straight linear interpolation between $\bar{U}(a_n^-)$ and $\bar{U}(a_n^+)$) and hence we can write it as $\bar{U}^{**}(x) = \
gamma_n x + \alpha_n$ for some $\gamma_n \in \mathbb{R}_{>0}$, $\alpha_n \in \mathbb{R}$, with $\{\gamma_n\}$ a sequence satisfying that if indices $n$ and $m$ are such that $a_{n}^{+} \le a_{m}^{-}$ then $\gamma_{n}\ge\gamma_{m}$. Thus, clearly $\bar U^{**}$ is differentiable in $A$.

Now, denote by $B$ the open interior of the set where $\bar{U}$ and $\bar{U}^{**}$ agree, i.e., $B := \{ x > \beta\, : \, \bar{U}(x) = \bar{U}^{**}(x)\}^\circ$. We will prove that, on the set $B$, the function $\bar{U}$ is continuously differentiable. Pick some point $x \in B$. Since $g$ is convex, it holds that $g'_r(x) \geq g'_l(x)$, where $g'_r, g'_l$ are the left- and right-hand derivatives, respectively. Thus, it follows by the differentiability of $U$ that $\bar{U}'_r(x) = U'\bigl(g(x)\bigr)) g'_r(x) \geq U'\bigl(g(x)\bigr)) g'_l(x) = \bar{U}'_l(x)$. But, on the other hand, the concavity of $\bar{U}^{**}$ implies $\bar{U}'_r(x) = \bigl(\bar{U}^{**}\bigr)'_r(x) \leq \bigl(\bar{U}^{**}\bigr)'_l(x) = \bar{U}'_l(x) $. Thus, the left- and right-derivatives have to agree for every $x \in B$.  We conclude, then, that the function is continuously differentiable there.

We then note that $\left( \overline A  \cup B \right)\backslash \beta = (\beta,\infty) $.
Thus to complete our argument, it remains only to prove continuous differentiability on one of the points $a\in \overline A \backslash \left(A\cup\beta\right)$. Note that for such $a$, we can find a sequence $a_{n_k}^{\pm}$ (assume without loss of generality it is $a_{n_k}^{-}$, as we can handle $a_{n_k}^{+}$ the same way), such that $\lim\limits_{k\to\infty}a_{n_k}^{-} =a$. Additionally, note that by continuity of $\bar{U}^{**}, \bar{U}$, and the fact that $a\not\in A$ it follows that $\bar{U}^{**}(a) = \bar{U}(a)$. Assume by contradiction that $\bar{U}^{**}$ is not continuously differentiable at $a$. It follows that
\begin{equation}\label{C1an}
\bar{U}'_r(a) \geq \bar{U}'_l(a) \geq  \bigl(\bar{U}^{**}\bigr)'_l(a) > \bigl(\bar{U}^{**}\bigr)'_r(a).
\end{equation}
The first inequality stems from the fact that (since $U$ is continuously differentiable) every point of non-differentia\-bi\-li\-ty of $\bar{U}$ is due to not having an interior derivative of $U \circ g$. However, for the convex function $g$ we have $g_l' \leq g_r'$ (and $U'\geq 0$). The strict inequality is the consequence of our assumption that $\bar{U}^{**}$ is not differentiable at $a$ and that it is a concave function.
The second inequality follows from the fact that $\bar{U}^{**}$ is the concave hull of $\bar{U}$ (and both functions agree on $a$). Indeed, using the concavity of $\bar{U}^{**}$ and the fact that $\bar U(a) =\bar U^{**}(a)$ we write
\[
\bar U'_l(a) \ge \lim\limits_{h\to0^{+}}\frac{\bar U(a) - \bar U(a - h)  }{h} \ge \lim\limits_{h\to0^{+}}\frac{\bar U^{**}(a) - \bar U^{**}(a- h)  }{h} = \bigl(\bar U^{**}\bigr)'_l(a).
\]
However, \eqref{C1an} leads to a contradiction, since, by a similar argument
\[
\bigl(\bar U^{**}\bigr)'_r(a) \ge \lim\limits_{h\to0^{+}}\frac{\bar U^{**}(a + h) - \bar U^{**}(a)  }{h} \ge \lim\limits_{h\to0^{+}}\frac{\bar U(a+h) - \bar U(a )  }{h} \ge \bar U'_r(a).
\]
Thus, $\bar{U}^{**}$ has to be continuously differentiable in $a$ and hence, on the whole interval $(\beta, \infty)$.

In passing we note that the differentiability of  $\bar{U}^{**}$ implies that $\bar{U}^*$ cannot be differentiable at any $\gamma_n$. Assume indirectly that it would be differentiable. Then, there exists some $\tilde{a} \in \R$ such that $-\bigl(\bar{U}^*\bigr)'(\gamma_n)= \tilde{a}$. Furthermore, convex duality implies $\gamma_n \in \partial \bar{U}^{**}(\tilde{a})$. However, the differentiability of $\bar{U}^{**}$ reduces the subdifferential to a singleton. This means that $\gamma_n$ can only be the slope of  $\bar{U}^{**}$ at the single point $\tilde{a}$, which is in contradiction to the fact that it is the slope on the whole interval $(a_n^-,a_n^+)$.

Finally to show b), we note that the strict convexity in the range of the gradient mapping is a classical consequence in convex Analysis, see e.g., \cite[Theorem E.4.1.2.]{HUL}, i.e., $\bar U^{*}$ is strictly convex on $\bigl\{ \bigl(\bar U^{**}\bigr)'(x)\, : \, x\in (\beta,\infty)\bigr\}$. We claim that $\bigl\{ \bigl(\bar U^{**}\bigr)'(x) \, : \, x\in (\beta,\infty)\bigr\} = \bigl(0,(\bar U^{**})'(\beta)\bigr)$. Indeed, $\bigl(\bar U^{**}\bigr)'$ is nonincreasing, and for $x>\max\{a_1^+,\beta\}$ 
\begin{align*}
\bigl(\bar{U}^{**}\bigr)'(x) & = \left\{ \begin{array}{ll} U'\bigl(g(x)\bigr)g'(x) & x \notin \overline{A} \\ \gamma_n & x \in [a_n^-,a_n^+] \end{array} \right\} \\
& = \left\{ \begin{array}{ll} U'\bigl(g(x)\bigr)g'(x) & x \notin \overline{A} \\ U'\bigl(g(a_n^-)\bigr)g'(a_n^-) & x \in [a_n^-,a_n^+] \end{array} \right\} \leq U'\bigl(g(x)\bigr),
\end{align*}
with $a_n^\pm$ the boundary points of intervals in $A$ as above. Thus, since $g$ is convex, nonconstant and nondecreasing function, it must satisfy $\lim\limits_{x \to \infty} g(x) = \infty$. It follows by the Inada condition at $\infty$ that $0 \leq \bigl(\bar{U}^{**}\bigr)'(\infty) \leq U'(\infty) = 0$. For the right hand of the domain of strict convexity of $\bar U^{*}$ we have to consider three cases. First, if $U(0)=-\infty$, then we have $\bigl(\bar{U}^{**}\bigr)'(\beta) = \infty$, since $\bar{U}(\beta) = \bar{U}^{**}(\beta) = -\infty$. We therefore obtain $\bigl\{ (\bar U^{**})'(x)\, : \,  x\in (\beta,\infty)\bigr\}=(0,\infty)$. Second, if $U(0)$ is real and $\bigl(\bar{U}^{**}\bigr)'(\beta) = \infty$, then we can conclude in similar manner that $\bigl\{ (\bar U^{**})'(x)\, : \, x\in (\beta,\infty)\bigr\}=(0,\infty)$. Finally, if $U(0)$ is real and $\bigl(\bar{U}^{**}\bigr)'(\beta)< \infty$, then $\bar U^{*}$ is strictly convex on $\bigl(0, \bigl(\bar{U}^{**}\bigr)'(0)\bigr)$. However,
 for $y\ge \bigl(\bar{U}^{**}\bigr)'(0)=\max_{x>\beta}
 \bigl(\bar{U}^{**}\bigr)'(x)$, we can conclude that $ \bar U(0) \le \sup_{x>\beta}\bigl( \bar U(x)-xy\bigr) \le\sup_{x>\beta}\bigl( \bar U^{**}(x)-xy\bigr)=  \bar U^{**} (0)$. Since $\bar U(0) = \bar U^{**} (0)$, it follows that $v(y)\equiv \bar U^{**} (0)$ on $\bigl[(\bar U^{**})'(\beta),\infty\bigr)$.
\qed
\end{proof}

Finally, we have to prove the dual asymptotic ellipticity of $\bar{U}^{**}$. The following result builds on and generalizes (in the one-dimensional case) the equivalence result of dual and classical asymptotic elasticity given by Deelstra, Pham and Touzi \cite[Proposition 4.1.]{DPhT} (their result can be seen as the linear case $g(x)=x$).
\begin{lemma}\label{AElemma}
The concavified function $\bar{U}^{**}$ satisfies the dual asymptotic elasticity condition \eqref{AEdual}, i.e.,
\[
AE^*\bigl(\bar{U}^{**}\bigr) = \limsup_{y \to 0} \sup_{x \in -\partial \bar{U}^*(y)} \frac{yx}{\bar{U}^*(y)} < \infty.
\]
\end{lemma}

\begin{proof}
First, we note that, by the slope bound and the non-constancy of $g$
\begin{equation*}
c:= \sup\bigcup_{x \geq 0} \partial g(x)
\end{equation*}
is finite and strictly positive. Thus, we obtain on one hand that there exists for every $\varepsilon>0$ some $x_0$ (which we will assume to be bigger then $\beta$) such that for all $x>x_0$
\begin{equation}\label{g-slope-bnds}
g(x_0) + (c-\varepsilon)(x-x_0) \leq  g(x) \leq g(0) + cx,
\end{equation}
and 
\[
( c-\varepsilon) \leq \inf_{[x_0,\infty)} \partial g \leq \sup_{[x_0,\infty)} \partial g \leq c.
\] 
Moreover, we note that, in the case of affine $\tilde{g}(x)=ax+b$ with $a \in \mathbb{R}_{>0}$ and $b \in \mathbb{R}$, we have that
\[
 \sup_{x \in \dom U \circ \tilde{g}}\Bigl( U \bigl(\tilde{g}(x)\bigr) - xy \Bigr) = \sup_{x > -\frac ba}\Bigl( U(ax+b) - xy \Bigr) = \sup_{z >0}\biggl( U(z) - z\frac{y}{a} \biggr) + \frac{by}{a} = U^*\Bigl(\frac{y}{a}\Bigr) + \frac{by}{a}. 
\]
Setting $a := c - \varepsilon$ and $b := g(x_0) - (c-\varepsilon) x_0$, we note that
\[
 \sup_{x \in \dom U \circ \tilde{g}}\Bigl( U \bigl(\tilde{g}(x)\bigr) - xy \Bigr) =  \sup_{x > x_0}\Bigl( U \bigl(\tilde{g}(x)\bigr) - xy \Bigr)
\]
as long as $y < U'\bigl(g(x_0)\bigr)(c-\varepsilon) =: y_0$. Thus, we can conclude by \eqref{g-slope-bnds} that for $y \in (0,y_0)$ it holds that
\begin{align*}
\bar{U}^*(y)  & = \sup_{x \in \dom \bar{U}} \Bigl(\bar{U}(x) - xy \Bigr) =  \sup_{x>\beta} \Bigl(\bar{U}(x) - xy \Bigr) \ge  \sup_{x>x_0} \Bigl(\bigl(U \circ g\bigr)(x) - xy \Bigr) \\
&\ge \sup_{x>x_0} \Bigl(\bigl(U \circ \tilde{g}\bigr)(x) - xy \Bigr) =   \sup_{x \in \dom U \circ \tilde{g}}\Bigl(\bigl(U \circ \tilde{g}\bigr)(x) - xy \Bigr) \\
&=  U^*\Bigl(\frac{y}{c-\varepsilon}\Bigr) + \frac{g(x_0)- (c-\varepsilon)x_0 }{c-\varepsilon}y.
\end{align*}
We note that from Lemma \ref{Uconcsmooth}, it follows that for $x>a_1^+$ (and all $x>\beta$ in the case of concave $\bar{U}$)
\begin{align}\label{Ustarprimebound}
\bigl(\bar{U}^{**}\bigr)'(x) & = \left\{ \begin{array}{ll} U'\bigl(g(x)\bigr)g'(x) & \quad x \notin \overline{A} \\ \gamma_n & \quad  x \in [a_n^-,a_n^+] \end{array} \right\} \nonumber \\
& = \left\{ \begin{array}{ll} U'\bigl(g(x)\bigr)g'(x) & \quad x \notin \overline{A} \\ U'\bigl(g(a_n^-)\bigr)g'(a_n^-) & \quad x \in [a_n^-,a_n^+] \end{array} \right\} \leq U'\bigl(g(x)\bigr)c.
\end{align}
By convex conjugacy we have
\[
x \in- \partial \bar{U}^*(y) \qquad \Longleftrightarrow \qquad y = \bigl(\bar{U}^{**}\bigr)'(x).
\]
Hence, since by concavity the gradient is nonincreasing, we see that
\[
x \leq-\inf  \partial \bar{U}^*(y) \qquad \Longleftrightarrow \qquad y \geq (\partial\bar{U}^{**})'(x).
\]
Thus, we can conclude that $(U')^{-1} = -(U^*)'$. Let $f^{-1}(y) := \inf \{z \, : \, f(z)>y\}$ be the generalized inverse of a function $f$. Then
\begin{align*}
AE^*\bigl(\bar{U}^{**}\bigr) & = \limsup_{y \to 0} \sup_{x \in -\partial \bar{U}^*(y)} \frac{yx}{\bar{U}^*(y)} =  \limsup_{y \to 0} \sup_{\{ x \, : \, y = (\bar{U}^*)'(x)\}} \frac{yx}{\bar{U}^*(y)} \\
& \leq\limsup_{y \to 0} \sup_{\{ x \, : \, y \leq (\bar{U}^*)'(x)\}} \frac{yx}{\bar{U}^*(y)} \leq \limsup_{y \to 0} \sup_{\{ x \, : \, y \leq  {U}'(g(x))c\}} \frac{yx}{\bar{U}^*(y)} \\
& =  \limsup_{y \to 0} \sup_{\{ x \, : \, -(U^*)'(y/c)  \geq  g(x) \}} \frac{yx}{\bar{U}^*(y)} \leq   \limsup_{y \to 0} \frac{yg^{-1}\bigl(-(U^*)'\bigl(\frac{y}{c}\bigr)\bigr)}{\bar{U}^*(y)} .
\end{align*}
We discern now two cases. First, if $-(U^*)'$ is bounded, then we can directly conclude that $AE^*\bigl(\bar{U}^{**}\bigr)<\infty$, since $\bar U^{*}(0)= U(\infty) >0$.  Second, if $-(U^*)'$ is unbounded, then by the Inada condition for $U$ we have that 
 \[
\limsup_{y \to 0}\, \bigl( -U^{*}(y)\bigr) =\limsup_{y\to 0} \, (U')^{-1}(y) =\infty.  
 \]
From \eqref{g-slope-bnds} we see that $y \le \frac{g(y) -g(x_0)}{c-\varepsilon} +x_0 $ holds for all $y\ge x_0$. Applying this to $y=g^{-1}(z)$ (note that $g$ is here a  true inverse since $x_0$ was assumed to be bigger then $\beta$) we conclude that $g^{-1}(z) \le \frac{z}{c-\varepsilon}-\frac{g(x_0)}{c-\varepsilon}+x_0$, for all $z > g(x_0)$. It follows that with $z=-(U^*)'\bigl(\frac{y}{c}\bigr)$ we have
\[
g^{-1}\Bigl(-(U^*)'\bigl(\frac{y}{c}\bigr)\Bigr) \le -\frac{1}{c-\varepsilon} (U^*)'\Bigl(\frac{y}{c}\Bigr)- \frac{(g(x_0) - (c-\varepsilon)x_0)}{c-\varepsilon}, 
\]
for $y$ satisfying $-(U^*)'\bigl(\frac{y}{c}\bigr)>g(x_0)$. By the unboundedness of $-(U^*)'$ this is satisfied for all $y$ small enough. Since $U$ satisfies the dual asymptotic elasticity condition by \cite[Proposition 4.1.]{DPhT} (cf. Remark \ref{AEequiv}) we have for some $M \in(0, \infty)$ that
\[
AE^*\bigl(U\bigr) =  \limsup_{y \to 0} \frac{-y(U^*)'(y)}{U^*(y)} < M < \infty.
\]
From
\[
U^*\Bigl(\frac{y}{c-\varepsilon}\Bigr) \geq U^*\Bigl(\frac{y}{c}\Bigr) + \varepsilon \frac{y}{c(c-\varepsilon)} \bigl(U^*\bigr)'\Bigl(\frac{y}{c}\Bigr),
\]
we conclude that 
\begin{align*}
AE^*\bigl(\bar{U}^{**}\bigr) \leq & \limsup_{y \to 0} \frac{yg^{-1}\bigl(-(U^*)'\bigl(\frac{y}{c}\bigr)\bigr)}{\bar{U}^*(y)}
\leq \limsup_{y \to 0} \frac{-\frac{y}{c-\varepsilon} (U^*)'\bigl(\frac{y}{c}\bigr)- \frac{(g(x_0) - (c-\varepsilon)x_0)}{c-\varepsilon}y}{U^*\bigl(\frac{y}{c-\varepsilon}\bigr) + \frac{g(x_0)- (c-\varepsilon)x_0 }{c-\varepsilon}y }\\
\leq & \limsup_{y \to 0} \frac{1}{c-\varepsilon}\frac{-y(U^*)'\bigl(\frac{y}{c}\bigr)}{U^*\bigl(\frac{y}{c-\varepsilon}\bigr)} + 1 \leq  \frac{1}{c-\varepsilon}\frac{1}{\frac{1}{M} -\frac{\varepsilon}{c(c-\varepsilon)}} + 1< \infty,
\end{align*}
for $\varepsilon>0$ chosen  small enough (note that $M$ only depends on the original utility function, hence it is independent of $\varepsilon$).
\qed
\end{proof}

We can now look more closely at how the concavified problem relates to the classical Kramkov/Schachermayer setting. The concavified utility function $\bar{U}^{**}$ is indeed continuously differentiable. It will follow from \eqref{Ustarprimebound} that it satisfies also the Inada condition $\bigl(\bar{U}^{**}\bigr)'(\infty) = 0$. Hence, by \cite[Proposition 4.1.]{DPhT} the primal asymptotic elasticity condition $AE\bigl( \bar{U}^{**}\bigr)<1$ is also satisfied. However, it fails, in general, the Inada condition $\bigl(\bar{U}^{**}\bigr)'(0) = \infty$. Furthermore, it will not necessarily be strictly concave.

Relying heavily on Proposition \ref{technicallemma}, we can now prove Theorem \ref{main}, which is the result concerning existence and uniqueness of an optimal solution of the dual problem \eqref{dualgen} as well as existence for the concavified problem \eqref{concavified}. In the next sections we will use this central result to discuss the uniqueness of the concavified problem as well as discuss how one can use the concavified problem to solve the original problem \eqref{original}.

\begin{proof}[Theorem \ref{main}]
It follows from Proposition \ref{technicallemma} that the conditions of Theorem \ref{BTZduality} are satisfied for the concavified utility function $\bar{U}^{**}$ with $\alpha = \beta$. This implies the finiteness and the duality statements of a).

The existence part of b) also follows directly from Theorem \ref{BTZduality}. For the uniqueness part, we note Proposition \ref{technicallemma} implies that $\bar{U}^*$ is strictly convex on $(0,(\bar{U}^{**})'(0))$. Assume, by contradiction, that $\hat{Y}_T^1(y)$ and $\hat{Y}_T^2(y)$ are the terminal values of two different optimizers of the dual problem such that 
\begin{align*}
\P\bigl[\{\hat{Y}_T^1(y) \neq \hat{Y}_T^2(y)\} \cap \{ \hat{Y}_T^1(y) \in(0,(\bar{U}^{**})'(0)) \} \cap \{\hat{Y}_T^2(y) \in(0,(\bar{U}^{**})'(0)) \}  \bigr]>0.
\end{align*}
That is, the random variables $\hat{Y}_T^1(y),\hat{Y}_T^2(y)$ differ on the set $(0,(\bar{U}^{**})'(0))$ with positive probability. It follows that for every $\lambda \in (0,1)$ and $Y^\lambda_T(y) := \lambda \hat{Y}^1_T(y) + (1-\lambda)\hat{Y}_T^2(y)$ we have, by the strict convexity of $\bar{U}^*$, that
\begin{align*}
\E \Bigl[\bar{U}^*\bigl( Y_T^\lambda(y)\bigr)\Bigr]& = \E \Bigl[\bar{U}^*\bigl( \lambda \hat{Y}_T^1(y) + (1-\lambda) \hat{Y}_T^2(y)\bigr)\Bigr] \\
&<  \lambda \E \Bigl[\bar{U}^*\bigl( \hat{Y}_T^1(y) \bigr)\Bigr]+ (1-\lambda)  \E \Bigl[\bar{U}^*\bigl( \hat{Y}_T^2(y)\bigr)\Bigr],
\end{align*}
which contradicts the optimality of $\hat{Y}_T^1(y)$, or $\hat{Y}_T^2(y)$.
 
To prove the remaining statements of a) we note that we have for every $\lambda \in (0,1)$ and $y_1$, $y_2 > 0$
\[
\lambda \hat{Y}(y_1)  + (1- \lambda)\hat{Y}(y_2)\in \lambda \mathcal{Y}(y_1) + (1-\lambda) \mathcal{Y}(y_2) =\mathcal{Y}\bigl(\lambda y_1 + (1- \lambda)y_2 \bigr).
\]
Thus, we can conclude by the strict convexity of $\bar{U}^*$ that for $\lambda \in (0,1)$, and $0<y_1<y_2 \leq \delta$ with $ \delta := \sup\{y>0 \, :\, \supp{(\hat{Y}_T(y))} \cap (0,(\bar{U}^{**})'(\beta)] \neq \emptyset\}$ we have that 
\begin{align*}
v\bigl(\lambda y_1 + (1-\lambda)y_2\bigr) &= \E \Bigl[ \bar{U}^*\bigl(\hat{Y}_T(\lambda y_1 + (1-\lambda)y_2)\bigr)\Bigr] \le\E \Bigl[ \bar{U}^*\bigl( \lambda\hat{Y}_T( y_1) + (1-\lambda)\hat{Y}_T(y_2)\bigr)\Bigr] \\ 
&< \lambda E \Bigl[ \bar{U}^*\bigl( \hat{Y}_T( y_1)\bigr)\Bigr] + (1-\lambda) E \Bigl[ \bar{U}^*\bigl(\hat{Y}_T(y_2)\bigr)\Bigr] = \lambda v(y_1) + (1-\lambda) v(y_2).
\end{align*}
Hence, $v$ is  strictly convex on $(0, \delta)$ and constant $\bar{U}^{**}(0)$ on $[\delta, \infty)$. By \cite[Theorem E.4.1.1.]{HUL} this implies the continuous differentiability of $w$.

Part c) follows directly from Theorem \ref{BTZduality} and the differentiability of $w$ in the interior of its domain, $(\beta, \infty)$.

Finally, d) is a direct consequence of Theorem \ref{BTZduality}.
\qed
\end{proof}

\section{The Original Problem: Wealth-independent Solution}\label{sec:7}

We are now finally ready to give the proof for Theorem \ref{general} and Proposition \ref{suffCrit}. We will rely heavily on the following results from the proof of Proposition \ref{technicallemma}, specifically Lemma \ref{Uconcsmooth}. The set $A$ where the two utility functions disagree is an open subset of $\Rplus$.  As such, $A$ is a countable union of pairwise disjoint open intervals,
\[
A := \bigcup_{n=1}^\infty (a^-_n,a^+_n) =  \bigl\{ x > 0 \, : \bar{U}(x) \neq \bar{U}^{**}(x)\bigr\}, \qquad a_n^-<a_n^+.
\]
 On every one of these intervals the function $\bar{U}^{**}$ is affine,  $\bar{U}^{**}(x) = \gamma_n x + \alpha_n$ for some $\gamma_n \in \mathbb{R}_{>0}$, $\alpha_n \in \mathbb{R}$, where $\{\gamma_n\}$ is a sequence satisfying that if indices $n$ and $m$ are such that $a_{n}^{+} \le a_{m}^{-}$ then $\gamma_{n}\ge\gamma_{m}$. We set
\[
\Gamma := \bigcup_{n=1}^\infty \bigl\{ \gamma_n \bigr\},
\]
and note that on every $\gamma_n$ the dual utility function $\bar{U}^*$ has a kink, i.e., the function $\bar{U}^*$ is not continuously differentiable. We insist that not every kink of $\bar{U}^*$ has to lie in $\Gamma$, nor is every region of linearity of  $\bar{U}^{**}$ necessarily contained in $A$ (e.g., when $\bar{U} =U \circ g$ is itself concave and has regions of linearity). However, by the duality relationship of $\bar{U}^{**}$ and $\bar{U}^{*}$, we know that, for the subdifferentials
\begin{equation}\label{dualag}
\bigl(\bar{U}^{**}\bigr)'(A) = \Gamma \qquad \mbox{and} \qquad -\partial \bar{U}^* (\Gamma)  \supseteq A
\end{equation}
holds true.

\begin{proof}[Theorem \ref{general}]
Given that $\hat{Y}_T(y)$ has a continuous law and is unique where $\partial \bar{U}^*$ is not vanishing, it follows that for any $f_1$, $f_2 \in -\partial \bar{U}^*$ we have
$f_1\bigl(\hat{Y}_T(y)\bigr) = f_2\bigl(\hat{Y}_T(y)\bigr)$ $\P$-a.s. Hence
\[
\hat{W}_T(x) = f \bigl(\hat{Y}_T \bigl(w'(x)\bigr)\bigr), \qquad -f \in \partial \bar{U}^*,
\] 
is $\P$-a.s. uniquely defined by a strictly increasing function $f$. Since $\hat{Y}_T\bigl(w'(x)\bigr)$ has a continuous law, so does $\hat{W}_T(x)$, proving a).

By the duality relationship \eqref{dualag} we can conclude that
\begin{align}
\P \bigl[\hat{W}_T(x) \in A \bigr] =& \P\bigl[f \bigl(\hat{Y}_T \bigl(w'(x)\bigr)\bigr) \in A \bigr]  \le \P \bigl[ \hat{Y}_T\bigr(w'(x)\bigr) \in \bigl(\bar{U}^{**}\bigr)'(A)\bigr] \label{eq:calc1} \\
= &\P \bigl[ \hat{Y}_T\bigr(w'(x)\bigr) \in \Gamma \bigr]  \leq \sum_{n=1}^\infty \P \bigl[ \hat{Y}_T\bigr(w'(x)\bigr) = \gamma_n \bigr] =0\nonumber,
\end{align}
since the distribution of $\hat{Y}_T(y)$ has no atoms for any $y>0$. Thus, $\hat{W}_T(x)$ is $\P$-a.s. equal to $0$ on $A$. Thus, we have on one hand 
\[
w(x) = \E\bigl[ \bar{U}^{**}\bigl(\hat{W}_T(x)\bigr)\bigr]  =   E \bigl[\bar{U}\bigl(\hat{W}_T(x)\bigr)\bigr] \leq  \sup_{X\in \mathcal{X}(x)} \E\bigl[ \bar{U}\bigl(X_T\bigr)\bigr]  =u(x),
\]
and on the other hand
\[
u(x) = \sup_{X \in \mathcal{X}(x)} \E\bigl[ \bar{U}\bigl(X_T\bigr)\bigr] \leq  \sup_{X \in \mathcal{X}(x)} \E\bigl[ \bar{U}^{**}\bigl(X_T\bigr)\bigr] = \E\bigl[\bar{U}^{**}\bigl(\hat{W}_T(x)\bigr)\bigr] = w(x).
\]
Thus, it is clear that $\hat{W}(x)$ is also an optimizer for the original problem, $\hat{X}(x) = \hat{W}(x)$, proving b).
\qed
\end{proof}

Note that we have said nothing about the optimal portfolio of the original problem {\em per se}, but only about the coincidence of its maximizer with that of the concavified problem. That is, the statement is as follows: when the law of the dual optimizer has no atoms, then there is no {\em 'biduality gap'}, and the original problem can be solved by considering the problem with the concavified utility function.

The following remark discusses the economic consequences of Theorem \ref{general}:

\begin{remark}
{${}$}
\begin{itemize}
\item[a)] The optimizer $\hat{X}(x)$ of Theorem \ref{main} satisfies $\hat{X}_T(x) \notin A$, $\P$-a.s. That is, the portfolio manager flees successfully all possible outcomes that underperform the concavification.
\item[b)] 
Similar to the calculation in \eqref{eq:calc1} we can show that the law of $\hat X_T(x)$ is atomless, except possibly an atom at $\beta$. Indeed, by Theorem \ref{general} it is enough to show that the distribution of $\hat W_T(x)$ is atomless, as it coincides with $\hat X_T(x)$ a.s. Take $z>\beta$ and $f \in -\partial \bar{U}^*$, then $\hat{W}_T(x) = f \bigl(\hat{Y}_T \bigl(w'(x)\bigr)\bigr)$ and
\begin{equation}
\P \bigl[\hat{W}_T(x) =z \bigr]= \P\bigl[f \bigl(\hat{Y}_T \bigl(w'(x)\bigr)\bigr) =z \bigr]  = \P \bigl[ \hat{Y}_T\bigr(w'(x)\bigr) = \bigl(\bar{U}^{**}\bigr)'(z)\bigr] =0.
\label{eq:calc2} 
\end{equation}
However, there is a possibility that an atom occurs at $z=\beta$. The same calculation shows that if $(\bar U^{**})'(\beta)=\infty$, then the distribution of $\hat X_T(x)$ cannot have an atom at $\beta$. Specifically, the law $\hat X_T(x)$ has an atom at $\beta$ if and only if $(\bar U^{**})'(\beta)<\infty$ and $\P\bigl[ \hat Y_T\bigl(w'(x)\bigr)\ge  (\bar U^{**})'(\beta)\bigr] >0$. Moreover, in this case, 
\[
\P[\hat X_T(x)=\beta]=\P\bigl[\hat Y_T\bigl(w'(x)\bigr)\ge (\bar U^{**})'(\beta)\bigr].
\]
This outcome, which occurs, for example, by pure call option payoffs in Black-Scholes markets with nonzero drift, is not very satisfactory for the investor, since the incentive scheme for the portfolio manager is such that the optimal strategy jeopardizes the whole capital with positive probability. What is worse, a call option incentive scheme leads to a higher probability of the ruin as the benchmark increases.
\item[c)] Carpenter \cite{carp}  also considers the case of a call option with random benchmark, $g(x )= (x-B_T)^+$. It is not to hard to integrate such options in our more general framework as long as $B_T \in L^\infty(\Omega, \mathcal{F}, \P)$, using the random endowment result of \cite[Theorem 3.2.]{BTZ}.
\end{itemize}
\end{remark}

\begin{proof}[Proposion \ref{suffCrit}] We know by Theorem \ref{main} that the value function of the dual problem can be represented as an infimum over equivalent local martingale measures, 
\begin{equation}\label{dualdensity}
v(y) = \inf_{\Q \in \mathcal{M}^e} \E \biggl[ \bar{U}^*\biggl(y \frac{d\Q}{d\P}\biggr)\biggr].
\end{equation}
Hence, we can, in particular, extract a sequence $Z^n \in \bigl\{\frac{d\Q}{d\P} \, : \, \Q \in \mathcal{M}^e\bigr\}$ so that $\E\bigl[\bar{U}^*(yZ^n)\bigr]$ converges to $v(y)$. Note that the sequence $Z^n$ is bounded in $L^1(\Omega, \mathcal{F},\P)$ , since the expectations of densities are bounded by one. Hence, we can apply Koml\'{o}s' Lemma (\cite[Theorem 4.27]{Bogachev}) to find a subsequence $Z^{n_k}$ and a random variable $Z$ such that every subsequence $Z^{n_{k_l}}$ of $Z^{n_k}$ converges to $Z$, $\P$-a.s. in the sense of Ces\`{a}ro. We note that $Z$ is a minimizer of \eqref{dualdensity} since
\[
\E\biggl[\bar{U}^*\biggl(\frac{y}{m} \sum_{j=1}^m Z^{n_{k_j}}\biggr)\biggr] \leq \frac{1}{m}  \sum_{j=1}^m \E\biggl[\bar{U}^*\Bigl(yZ^{n_{k_j}}\Bigr)\biggr] .
\]
By the convexity of $\bar{U}^*$, the right hand converges as Ces\`{a}ro-subsequence of a convergent sequence to $v(y)$. Whereas the convex combination of the random variables on the left hand is the density corresponding to some equivalent local martingale measure by the convexity of $\mathcal{M}^e$. 

Next, we assert that $Z$ has a distribution, which has a continuous law. Indeed, since the laws of all the approximating $Z^n$ are uniformly absolutely continuous with respect to Lebesgue measure, so are the approximating Ces\`{a}ro sums. Denote these sums by $\tilde Z^n$. Uniform absolute continuity with respect to the Lebesgue measure of the laws of $\tilde Z^n$ implies that, for the respective cumulative distribution functions, it holds that for every  $\varepsilon>0$ and all $t \in\R$ there exists a $\delta = \delta(\varepsilon) >0$ such that $\sup_n \sup_{t \in \R} \vert F_{\tilde Z^n}(t + \delta) - F_{\tilde Z^n}(t)\vert < \varepsilon$. We have that $\tilde Z^n \to Z$ in distribution, so $F_{\tilde Z^n} \to F_Z$ at all  points of continuity of the cumulative distribution function $F_Z$. To prove our assertion, it is enough to show that $F_Z(x)$ is continuous for every  $x\in\R$. Indeed, since $F_Z$ is increasing and bounded, it has at most countable number of discontinuity points. Take for given $\
varepsilon>0$ some $x_1$, $x_2 \in \mathbb{R}$, $x_1<x<x_2$, such that $x_2-x_1< \delta \bigl(\frac{\varepsilon}{3}\bigr)$, and such that $F_Z$ is continuous at both, $x_1$ and  $x_2$. Then $F_{\tilde{Z}^n}(x_2)-F_{\tilde{Z}^n}(x_1)<\frac{\varepsilon}{3}$ for all $n\in\N$. We can also chose $n$ big enough such that $\vert \, F_{\tilde{Z}^n}(x_i) - F_Z(x_i) \, \vert<\frac{\varepsilon}{3}$, $i=1,2$. Finally, we can conclude that, for all $y \in [x_1,x_2]$
\[
\bigl\vert \, F_Z(x) - F_Z(y) \,\bigr\vert \le  F_Z(x_2)- F_Z(x_1) \le F_{\tilde{Z}^n}(x_2) - F_{\tilde{Z}^n}(x_1) + \frac{2\varepsilon}{3} < \varepsilon.
\]
Thus, $F_Z$ is continuous at $x$.
\qed
\end{proof}

\begin{remark}
The proof becomes even simpler if one switches to the more abstract level of the bipolar theorem on $L^0_+(\Omega, \mathcal{F}, \P)$ of \cite{BS}. The set of nonnegative random variables dominated by the terminal values of the processes in $\mathcal{Y}(y)$ is the bipolar of $\bigl\{y \frac{d\Q}{d\P} \, : \, \Q \in \mathcal{M}^e\bigr\}$, i.e., the smallest solid, convex set closed in the sense of convergence in probability that contains $\bigl\{y \frac{d\Q}{d\P} \, : \, \Q \in \mathcal{M}^e\bigr\}$.  Thus, every element in this set is given as a limit of $y$ times a Radon-Nikod\'{y}m derivative. Thus, by Riesz's theorem, we can extract a subsequence, which converges almost surely. Moreover, in this abstract perspective we are able to give the following interpretation. The optimizer of utility maximization under a convex incentive scheme is well-behaved (i.e., atomless) if the whole set of possible optimizers is well-behaved.  Furthermore, this set is (up to a multiplicative factor) simply the bipolar of the 
set of Radon-Nikod\'{y}m derivatives of equivalent local martingale measures. Thus, if this set is nice enough (i.e., the distribution of its elements are uniformly absolute continuous with respect to the Lebesgue measure), we always obtain a unique optimizer for utility maximization under convex incentive schemes, independent of the initial capital and the concrete choice of the incentive scheme.
\end{remark}

\section{The Original Problem: Wealth-dependent Solution}\label{sec:8}

Inspired by Example \ref{counter} and, specifically by the case with zero drift in section \ref{ex:zeroDrift}, we try now to deduce how one can extend Theorem \ref{main} to get existence and/or uniqueness results for particular initial conditions. For $y>0$ we denote by
\[
\Delta(y) = \bigl\{\delta >0 \, : \,  \P\bigl[\hat{Y}_T(y) = \delta\bigr] >0 \bigr\}
\]
the at most countable set of atoms of the law of the dual optimizer $\hat{Y}_T(y)$. Moreover, we recall the notations 
\[
A = \bigcup_{n=1}^\infty (a^-_n,a^+_n) =  \bigl\{ x > 0 \, : \bar{U}(x) \neq \bar{U}^{**}(x)\bigr\}, \qquad \Gamma = \bigcup_{n=1}^\infty \bigl\{ \gamma_n \bigr\},
\]
where $\gamma_n$ is the slope of $\bar{U}^{**}$ on $(a_n^-, a_n^+)$. We are now able to make the following statement.

\begin{theorem}\label{concrete}
The optimizer $\hat{W}(x)$ for the concavified problem \eqref{concavified} is unique for $x>\beta$ if
\begin{equation}\label{uniccond}
\Delta\bigl(w'(x)\bigr) \cap \Gamma = \emptyset.
\end{equation}
Moreover, in this case, $\hat{X}(x) = \hat{W}(x)$ is the unique solution to the original problem \eqref{original}.
\end{theorem}

\begin{proof}
First, note that condition \eqref{uniccond}  implies that  no atom of the distribution of $\hat{Y}_T\bigl(w'(x)\bigr)$ lies on a point in the domain of $\bar{U}^*$ where this function is not differentiable. Thus, we can conclude, as in the proof of Theorem \ref{general}, that, for $f_1$, $f_2 \in - \partial \bar{U}^*$, we have
$f_1\bigl(\hat{Y}_T(w'(x))\bigr) = f_2\bigl(\hat{Y}_T(w'(x))\bigr)$, $\P$-a.s. Hence
\[
\hat{W}_T(x) = f \bigl(\hat{Y}_T \bigl(w'(x)\bigr)\bigr), \qquad f \in -\partial \bar{U}^*,
\] 
is $\P$-a.s. uniquely defined by a strictly increasing function $f$, which proves uniqueness of $\hat{W}(x)$. To prove the existence of an optimizer of the original problem, we note that, from \eqref{uniccond}, we know that $\P \bigl[ \hat{Y}_T\bigr(w'(x)\bigr) = \gamma_n \bigr] =0$. Thus, similar to the proof of Theorem \ref{general}, we can conclude that $\hat{X}_T(x)$ is (the unique) solution to the original problem.  
 \qed
\end{proof}

For the case that $x > \beta$ such that $ \Delta\bigl(w'(x)\bigr) \cap \Gamma \neq \emptyset$, we cannot generally recover any of our results. In particular:
\begin{itemize}
\item[a)] The optimizer of the concavified problem may not be unique, as discussed in the remark at the end of Example \ref{counter}, Case 2.
\item[b)] It can happen that the optimum of the concavified problem is not reached by the value function of the original problem, i.e., $u(x) < w(x)$. An example therefore will be given below in Example \ref{incompcounter}.
\item[c)] Even if the maximum of the concavified problem can be reached by the original value function, i.e., $u(x) = w(x)$, it may happen that the optimizer of the original problem is not unique. To see this, we use the setting of Example \ref{counter} (with initial capital $1$), changing only the incentive scheme
\[
\check{g}(x) = \left\{ \begin{array}{ll} \frac{x^2}{24} &\quad 0 \leq x \leq 6, \\ \frac{1}{2}(x-3) & \quad x>6,\end{array} \right.
\]
which is a convex function with slope bounded by one. However, $U \circ \check{g} = \bar{U}^{**}$. Thus, all of the solutions of the concavified problem in Example \ref{counter} are also solutions to the original problem with incentive scheme $\check{g}$.
\end{itemize}

\begin{example}\label{incompcounter}
To see that the optimizer of the concavified problem can be strictly bigger then any admissible terminal value for the original problem, we once again use the utility function and incentive scheme of \eqref{eq:U-example} from Example \ref{counter}, namely $U(x) = 2\sqrt{x}$ and $g(x) = \frac14(x-3)^+$. We also take $x=1$ as initial capital. To describe the discounted stock price process, we fix an $(\Omega, \mathcal{F}, \P)$-measurable random variable $R$ that satisfies $\P[R=2] =\P[R=1/2] = 1/2$ and consider the process
\[
 S_t = \left\{ \begin{array}{ll} 1 & \quad0 \leq t < T/2, \\ R & \quad T/2 \leq t\leq T, \end{array} \right.
\]
in its natural filtration. Thus, in essence, our model a disguised form of a binomial model. We note that
\[
-\partial \bar U^{*}(y) = \left\{  \begin{array}{ll}
\frac1{4y^2}+3& \quad 0<y<\frac{\sqrt{3}}6,\\
 \left[0, \right]  & \quad y=\frac{\sqrt{3}}6, \quad\\
0& \quad y > \frac{\sqrt{3}}6,
\end{array}
\right.
\]
and $\mathcal{M}^e=\{\Q\}$, where the measure $\Q$ is given via the Radon-Nikod\'{y}m derivative
\[
Z_T\define\left. \frac{d \Q}{d\P} \right\vert_{\mathcal{F}_T} = \frac23 \ind_{\{S_T=2\}}+\frac43\ind_{\{S_T=\frac12\}},
\]
implying  $\Q[R=2] = 1/3$ and $\Q[R=1/2] = 2/3$.

Our goal is to show that $u(1)<w(1)$. To compute $u(1) = \sup_{X \in \mathcal{X}(1)} \E\bigl[\bar U\bigl(X_T\bigr)\bigr]$, we note that, for any predictable $S$-integrable investment strategy $H$	
\[
X_T^{1,H} = x + \int_0^T H_t\, dS_t = 1+H_{T/2} \Bigl(S_{T/2}-S_{T/2-}\Bigr) = \left\{ 
 \begin{array}{ll} 
 1+H_{\frac{T}{2}} & R=2,\\
 1-\frac{H_{T/2}}2 & R=1/2.
 \end{array}\right.
\]
Since $X^{1,H} \in\mathcal X(1)$ has to be nonnegative, it follows that $H_{T/2} \in [-1, 2]$. Hence, $0\le X_T^{1,H}\le 3$, and we can conclude that 
\[
u(1)=\sup_{H} \E\bigl[\bar U \bigl(X_T^{1,H}\bigr)\bigr]=0. 
\]

For the calculation of $w(1)$ we use the fact that, in a complete market, $\mathcal{M}^e=\{\Q\}$.  Thus, the dual value function can be directly computed via the unique dual optimizer $\hat Y_T(y) = yZ_T$,
\[
v(y) = \inf_{Q \in \mathcal{M}^e} \E\biggl[ \bar U^*\biggl(y \frac{dQ}{dP}\biggr)\biggr]=\E\Bigl[\bar U^{*}(yZ_T)\Bigr]=\left\{  \begin{array}{ll}
\frac9{32y}-3y& \quad 0<y\le \frac{\sqrt{3}}8,\\
\frac3{16y}-y& \quad \frac{\sqrt{3}}8 < y <\frac{\sqrt{3}}4,\\
0& \quad y\ge \frac{\sqrt{3}}4.
\end{array}
\right.
 \]

Now, calculating the subdifferential,
\[
-\partial v(y) = \left\{  \begin{array}{ll}
\frac9{32y^2}+3& \quad 0<y<\frac{\sqrt{3}}8,\\
\left[5 , 9\right]  & \quad y=\frac{\sqrt{3}}8,\\
\frac3{16y^2}+1& \quad \frac{\sqrt{3}}8< y<\frac{\sqrt{3}}4,\\
\left[0,2\right]& \quad y=\frac{\sqrt{3}}4,\\
0& \quad y\ge \frac{\sqrt{3}}4,
\end{array}
\right. 
\]
and using by convex duality that $y = w'(x)$ if and only if $x\in-\partial v(y)$, we conclude that for $x=1$ it follows that $w'(1)=\sqrt{3}/4$. Thus, Theorem \ref{BTZduality} implies that
\begin{align*}
\hat W_T(1)\in-\bigl(\partial \bar U^{*}\bigr)\Bigl(\hat Y_1\bigl(w'(1)\bigr)\Bigr) & =-\bigl(\partial \bar U^{*}\bigr)\biggr(\frac{\sqrt{3}}6  \ind_{\{S_T=2\}}+\frac{\sqrt{3}}{3}\ind_{\{S_T=\frac12\}}  \biggr)\\
& =[0,6] \ind_{\{S_T=2\}}+\{0\}\ind_{\{S_T=\frac12\}} 
\end{align*}
and we can conclude by the admissibility constraint $\E \bigl[\hat W_T(1)\hat Y_1\bigl(w'(1)\big)\bigr]=w'(1)$ that 
\[
\hat{W}_T(1)=3 \ind_{\{S_T=2\}}+0\ind_{\{S_T=\frac12\}}. 
\]
This can be seen also in a simpler way. Since $X_T^{1,H} = x + \int_0^T H_t\, dS_t $ depends only on $H_{T/2}$ which, by predictability, has to be $\mathcal{F}_{T/2-} = \mathcal{F}_0$-measurable and hence constant. We have by admissibility $-1\le H\le2$. Hence
 \[
 w(1)=\sup_{H}\E\Bigl[\bar U^{**}\bigl(X_T^{1,H}(1)\bigr)\Bigr] = \sup_{H} \frac{\sqrt{3}}6\Bigl(\P[R=2](1+H)+\P[R=1/2](1-H/2)\Bigr)=\frac{\sqrt{3}}4.
 \]
The maximum is achieved with $H=2$, i.e., the optimal portfolio is $\hat{W}_T(1)=3 \ind_{\{S_T=2\}}+0\ind_{\{S_T=\frac12\}}$. It follows in either case that $w(1)=\frac{\sqrt{3}}4$. Thus, we conclude that $0=u(1)<w(1)=\frac{\sqrt{3}}4$.
\end{example}

Note, finally, that such behavior can be excluded in the case of complete markets; Reichlin \cite[Section 5]{Rei} shows that the optimizer of the concavified problem is an optimizer of the original problem if the underlying probability space is atomless.

\section{Conclusion}\label{sec:conclusion}

We have considered the non-concave utility maximization problem as seen from the point of view of a fund manager, who manages the capital for an investor and who is compensated by a convex incentive scheme.  We have proved the existence and uniqueness of the dual optimizer and also proved the existence and uniqueness of the original problem for arbitrary initial capital in case in which the dual optimizer has a continuous distribution. We have shown that this is true in a large class of (possibly incomplete) market models, independent of the specific incentive scheme. When this condition fails, we have proved the existence of a unique solution for the concavified problem and shown that this solution is also  a solution of the original problem under additional assumptions on the initial capital. However, there are models, where, for some initial capital, the optimal value of the concavified problem cannot be reached, as we have demonstrated through a counterexample. Moreover, we have illustrated our findings 
by specific examples, which in essence contain the explicit solution strategies for complete markets. Finally, we have discussed the economic implications of our findings.

\begin{acknowledgements}
Both authors acknowledge partial financial supported by NSF grant DMS-0739195 and want to thank Ren\'{e} Carmona for suggesting the problem under consideration and steady encouragement.They are thankful to two anonymous referees and an associate editor for thoughtful remarks and comments that have improved the quality of the article. Also thanks to Gerard Brunick, Christian Reichlin, Ronnie Sircar and Ramon van Handel for helpful discussions and comments. Thanks to Matt Lorig, who put a lot of effort to improve the readability and grammatical correctness of the text.
\end{acknowledgements}

\bibliographystyle{spmpsci}
\bibliography{bibagent}
\end{document}